\documentclass[a4paper,reqno,11pt]{amsart}

\usepackage[hmargin=3cm,vmargin=3cm]{geometry}

\usepackage{amssymb,amsthm,mathtools,cancel,mathrsfs,mathdots,color,graphicx,enumerate}

\usepackage[colorlinks=true, pdfstartview=FitV, urlcolor=blue, citecolor=red, linkcolor=blue]{hyperref}
\usepackage{stmaryrd}
\SetSymbolFont{stmry}{bold}{U}{stmry}{m}{n} 

\usepackage[T1]{fontenc}
\usepackage[landscape=true]{pdflscape}

\def\res{\mathop{\mathrm {res}}\limits_}
\def\e{\mathrm{e}}
\def\g{\gamma}
\def\tr{\mathrm{tr}\,}
\def\Im{\mathrm{Im}\,}
\def\Re{\mathrm{Re}\,}

\def\f{\mathbf{f}}
\def\g{\mathbf{g}}

\def\be{\begin{equation}}
\def\ee{\end{equation}}

\newtheorem{theorem}{Theorem}[section]

\newtheorem{assumption}[theorem]{Assumption}

\newtheorem{lemma}[theorem]{Lemma}
\newtheorem{proposition}[theorem]{Proposition} 
\newtheorem{corollary}[theorem]{Corollary}
\newtheorem{RH}[theorem]{Riemann--Hilbert problem}

\newtheorem{remark}[theorem]{Remark}

\begin{document}

\numberwithin{equation}{section}

\title[Bessel kernel determinants and integrable equations]{Bessel kernel determinants\\ and integrable equations}

\author{Giulio Ruzza}
\address{Departamento de Matem\'atica, Faculdade de Ci\^encias da Universidade de Lisboa, Campo Grande Edif\'{i}cio C6, 1749-016, Lisboa, Portugal}
\email{gruzza@fc.ul.pt}

\date{}

\begin{abstract}
We derive differential equations for multiplicative statistics of the Bessel determinantal point process depending on two parameters.
In particular, we prove that such statistics are solutions to an integrable nonlinear partial differential equation describing isospectral deformations of a Sturm--Liouville equation.
We also derive identities relating solutions to the integrable partial differential equation and to the Sturm--Liouville equation which imply an analogue for Painlev\'e\,V of Amir--Corwin--Quastel ``integro-differential Painlev\'e\,II equation''.
This equation reduces, in a degenerate limit, to the system of coupled Painlev\'e\,V equations derived by Charlier and Doeraene for the generating function of the Bessel process, and to the Painlev\'e\,V equation derived by Tracy and Widom for the gap probability of the Bessel process.
Finally, we study an initial value problem for the integrable partial differential equation.
The approach is based on Its--Izergin--Korepin--Slavnov theory of integrable operators and their associated Riemann--Hilbert problems.
\end{abstract}


\maketitle

\section{Introduction and results}

The \textit{Bessel process} is the determinantal random point process on $(0,+\infty)$ with correlation kernel given by the \textit{Bessel kernel}
\be
\label{eq:Besselkernel}
K_{\mathsf{Be}}(\lambda,\mu) = \frac 12\frac{\mathrm{J}_\alpha(\sqrt\lambda)\sqrt \mu\,\mathrm{J}_{\alpha}'(\sqrt \mu)-\sqrt \lambda\,\mathrm{J}_\alpha'(\sqrt \lambda)\mathrm{J}_{\alpha}(\sqrt \mu)}{\lambda-\mu},\qquad \lambda,\mu>0.
\ee
(When $\lambda=\mu$, this has to be taken in the sense of the limit.)
Here, $\mathrm{J}_\alpha(z)$ is the Bessel function of the first kind of order $\alpha$ and argument $z$, $\mathrm{J}_\alpha'(z)=\tfrac{\mathrm{d}\mathrm{J}_\alpha(z)}{\mathrm{d} z}$, and $\alpha>-1$.
The Bessel process belongs to a family of determinantal processes arising in Random Matrix Theory to describe certain \textit{universal} large-size scaling limits; other prominent members of this family are the \textit{Sine} and \textit{Airy processes}.
These processes also make their appearance in several interacting particle models beyond random matrices (notably, trapped non-interacting fermions and random partitions).

In this work, we continue the systematic study of multiplicative expectations of the determinantal point processes defined by these correlation kernels, as well as of associated ``positive-temperature'' deformations of such kernels.
In particular, this work parallels~\cite{CCR,CT}, devoted to the Airy and Sine processes.
The aforementioned positive-temperature deformations have recently appeared in models of trapped free fermions at positive temperature~\cite{DLDMS,LDMRS,LLDMS}, in Random Matrix Theory~\cite{J,LW,GS,BLEdge,BLBulk}, and, specifically for the Airy kernel, in connection with the Kardar--Parisi--Zhang equation with narrow-wedge initial data~\cite{ACQ,BG}.

More precisely, we will always assume that $\alpha$ is real and larger than $-1$, that $x$ is real and positive, and that $t$ is real, and we study the Fredholm determinant
\be
\label{eq:Q}
Q_\sigma(x,t)=1+\sum_{n=1}^{+\infty}\frac{(-1)^n}{n!}\int_{(0,+\infty)^n}\det\bigl(K_{\mathsf{Be}}(\lambda_i,\lambda_j)\bigr)_{i,j=1}^n\prod_{i=1}^n\sigma(x^{-2}\lambda_i+t)\,\mathrm{d} \lambda_i
\ee
under the conditions on the function $\sigma$ detailed in~Assumption~\ref{assumption} below.

The study is motivated by the following remarkable probabilistic content of the Fredholm determinant $Q_\sigma(x,t)$ defined in~\eqref{eq:Q}, which we will review in Section~\ref{sec:probability}.

\begin{itemize}
\item $Q_\sigma(x,t)$ is a multiplicative statistics of Bessel process.
In particular, it generalizes the gap probability and, more generally, the generating function of the Bessel process.
Our results reduce, by appropriate limits which we will detail below, to the celebrated results of Tracy and Widom~\cite{TW}, expressing the gap probability in terms of a particular solution of (a special case of) the \textit{Painlev\'e\,V equation}, and of Charlier and Doeraene~\cite{CD}, similarly expressing the generating function in terms of a particular solution to a \textit{system of coupled Painlev\'e\,V equations}.
\item $Q_\sigma(x,t)$ is a full gap probability of a \textit{thinned} Bessel process.
\item $Q_\sigma(x,t)$ is the smallest particle distribution of a \textit{positive-temperature deformation} of the Bessel process.
\end{itemize}

Asymptotic properties (as $x\to+\infty$) of these Fredholm determinant have been investigated in~\cite{BE}. (Cf.~Section~\ref{sec:PT}.)

\begin{assumption}
\label{assumption}
We have $\sigma=\sigma_0+\sigma_1$ where $\sigma_0,\sigma_1:\mathbb R\to\mathbb R$ are such that
\begin{itemize}
\item $\sigma_0$ is infinitely differentiable everywhere, and
\item $\sigma_1$ is piecewise constant with finitely many discontinuities.
\end{itemize}
Moreover, $r^a\sigma(r)\to 0$ as $r\to+\infty$, for all $a>0$ and $\sigma'_0\in L^1(0,+\infty)$.
\end{assumption}

Under these assumptions, $\sigma$ is in $L^\infty(t,+\infty)$ and is of bounded variation on $[t,+\infty)$, for any $t\in\mathbb R$.
In particular, we will denote 
\be
\label{eq:dsigma}
\int_A h(\lambda)\mathrm{d}\sigma(\lambda) = \int_A h(\lambda)\sigma'_0(\lambda)\mathrm d\lambda + \sum_{j=1}^k1_{r_j\in A}m_jh(r_j),\quad\mbox{for any Borel set $A\subseteq\mathbb R$,}
\ee
the Riemann--Stieltjes integral with respect to $\sigma$ (of a continuous function $h:A\to\mathbb C$).
Here, we introduce the notation $r_1,\dots,r_k$ for the discontinuity points of $\sigma$ and $m_j=\lim_{\epsilon\to 0_+}\bigl(\sigma(r_j+\epsilon)-\sigma(r_j-\epsilon)\bigr)$ the corresponding jumps.

\medskip

Let $v=v( x , t )$ be defined by
\be
\label{eq:defvstatement}
v=-(\log Q_\sigma)_x +\frac{ x  t }{2}-\frac{4\alpha^2-1}{8x}.
\ee
(We use subscripts to denote differentiation.)
Clearly, $v$ also depends on $\sigma$, though we will omit to indicate it explicitly.

\begin{theorem}\label{thm}
Under the hypotheses in Assumption~\ref{assumption}, the following facts hold true.
\begin{itemize}
\item[(i)] For every $t>0$, there exists $c_t>0$ such that $Q_\sigma(x,t)\not=0$ for $0<x<c_t$.
If we assume in addition that $0\leq\sigma\leq 1$, then $Q_\sigma(x,t)\not=0$ for all $x>0,\ t\in\mathbb R$.
\item[(ii)]
For all $x>0,\ t\in\mathbb R$ such that $Q_\sigma(x,t)\not=0$, the function $v=v(x,t)$ satisfies the nonlinear partial differential equation
\be
\label{eq:PDEstatement}
(2v_x- t)\,{v_t}^2+\frac 14{v_{xt}}^2-\frac 12v_{xxt}\,v_t\,=\,\frac{\alpha^2}4 .
\ee
\item[(iii)]
For all $\lambda>t$, there exists $f=f( x ;\lambda, t )$ which solves the boundary value problem
\be
\label{eq:Schrodingerstatement}
\begin{cases}
-f_{xx}\,+\,2\,v_x\,f\,=\,\lambda\, f,& x \in (0,+\infty),
\\
f\sim \sqrt x\,\mathrm{J}_\alpha(x\sqrt{\lambda-t}),& x \to 0_+ .
\end{cases}
\ee
In terms of this function, we have
\be
\label{eq:integration}
Q_\sigma(x,t) = \exp\int_0^x\log\biggl(\frac xy\biggr) \biggl(\int_ t ^{+\infty}(\lambda- t)\,f( y;\lambda, t )^2\,\mathrm{d}\sigma(\lambda)\biggr)\mathrm{d} y .
\ee
Moreover, $f$ also satisfies
\be
\label{eq:eqtstatement}
f_t\,=\,\frac 1{\lambda-t}\biggl(\frac 12 v_{xt}\,f\,-\,v_t\,f_x\biggr)
\ee
and
\be
\label{eq:nonlocal}
2v_t=x+\int_t^{+\infty}f(x;\lambda,t)^2\mathrm{d}\sigma(\lambda).
\ee
\end{itemize}
\end{theorem}

\begin{remark}
It will follow from the construction of $f(x;\lambda,t)$ that, for fixed $x,t$,
\be
\label{eq:asympfintegrals}
f(x;\lambda,t)=O\bigl(\lambda^{-\frac 14}\bigr),\ \ \lambda \to +\infty,\quad\mbox{and}\quad
f(x;\lambda,t)=O\bigl((\lambda-t)^{\frac\alpha 2}\bigr),\ \ \lambda\to t_+, 
\ee
cf.~Remark~\ref{remark:asymp}.
Therefore, all the integrals over $\lambda$ in $(t,+\infty)$ appearing in this theorem are well-defined.
\end{remark}

The proof is given in sections~\ref{sec:RH}--\ref{sec:proof}.
Next, we give some comments on Theorem~\ref{thm}.

\subsection{Remarks on Theorem~\ref{thm}}
\subsubsection{Uniqueness of \texorpdfstring{$f$}{f} for \texorpdfstring{$\alpha\geq 0$}{alpha positive}}
Given $f=f(x;\lambda,t)$ as in Theorem~\ref{thm}, we can write the general solution to $-g_{xx}+2v_x g=\lambda g$ as $g=f\bigl(A+B\int_{x_0}^xf(y)^{-2}\mathrm{d} y\bigr)$ depending on two arbitrary complex constants $A$ and $B$ (provided we fix $x_0$).
The boundary condition for~$f$ as~$x\to 0_+$ implies that $f(x)$ has no zeros for $x$ sufficiently small, such that the expression is well defined for $0<x<x_0$ with $x_0$ sufficiently small.
Then, as long as $\alpha\geq 0$, the boundary condition for~$f$ also implies that $\int_{x_0}^xf(y)^{-2}\mathrm{d} y$ diverges as $x\to 0_+$ which shows that the boundary value problem~\eqref{eq:Schrodingerstatement} has a unique solution.

\subsubsection{Integrability}
Let $v=v(x,t)$ be any solution to~\eqref{eq:PDEstatement}.
Taking one $x$-derivative we obtain that, if $v_t\not=0$, then
\be
\label{eq:differintro}
v_{xx}v_t+(2v_x-t)v_{xt}-\frac 14 v_{xxxt}=0.
\ee
(In case $v$ is given by~\eqref{eq:defvstatement}, we will prove that~\eqref{eq:differintro} holds true in any case, cf.~\eqref{eq:differ} below).
The integrability of~\eqref{eq:differintro}, and thus of the potential equation~\eqref{eq:PDEstatement}, can be highlighted by rewriting it as
\be
\label{eq:zerocurvatureintro}
\mathscr L_t(\mathscr L-t) = [\mathscr B,\mathscr L],\qquad
\mathscr L=-\partial_x^2+2v_x,\quad \mathscr B=-v_t\partial_x+\frac 12 v_{xt}.
\ee
(Note that $\mathscr L_t$ is the operator of multiplication by $2v_{xt}$.)
Then,~\eqref{eq:zerocurvatureintro} is the compatibility of the \textit{Lax pair} formed by the linear equations $\mathscr L f=\lambda f$ and $(\lambda-t)f_t=\mathscr B f$, which are exactly~\eqref{eq:Schrodingerstatement} and~\eqref{eq:eqtstatement}.

Let us also point out that letting
\be
\mathscr R_u=-\frac 14\partial_x^2+2u+u_x\partial_x^{-1}
\ee 
be the~\textit{Korteweg--de\,Vries recursion operator}~\cite{Olver}, the equation~\eqref{eq:differintro} can be recast as
\be
(\mathscr R_u-t)u_t=0,\qquad u=v_x.
\ee
This shows a similarity with the \emph{negative Korteweg--de\,Vries equation} $\mathscr R_\omega\omega_t=0$ (see~\cite{V}) for a function $\omega=\omega(x,t)$ but we were unable to relate them further.

\subsubsection{Special solutions}
Equation~\eqref{eq:PDEstatement} admits special solutions in terms of Bessel functions, explicitly given by
\be
v(x,t) = -\partial_x\log\det(M_{ij})_{i,j=1}^m+\frac {xt}2-\frac{4\alpha^2-1}{8x},\quad
M_{ij}=x^2 K_{\mathsf{Be}}\bigl(x^2(\mu_i-t),x^2(\nu_j-t)\bigr),
\ee
depending, for $m\geq 1$, on $2m$ real parameters $\mu_1,\nu_1,\dots,\mu_m,\nu_m$.
These can be derived by \textit{Darboux transformations} of the trivial solution $v(x,t)=\frac {xt}2-\frac{4\alpha^2-1}{8x}$ (corresponding to $\sigma=0$).
We will not derive these solutions here and rather refer to~\cite{CGRT} where Darboux transformations are applied to the corresponding equations for the Airy process rather than the Bessel process.
It is worth pointing out that the methodology of loc. cit. (see also~\cite{CG}) can be used to show that Darboux transformations of nontrivial solutions $v(x,t)$ (i.e., given by~\eqref{eq:defvstatement}, with nonzero $\sigma$) also enjoy a probabilistic interpretation, in terms of \textit{J\'anossy densities} of thinned Bessel processes.

Moreover, equation~\eqref{eq:PDEstatement} also admits \textit{self-similar} solutions in the form
\be
\label{eq:selfsimilar}
v(x,t)=\frac 1x k(\zeta)+\frac{xt}2-\frac{4\alpha^2-1}{8x},\qquad \zeta=x^2(r-t),
\ee
where $r$ is an arbitrary real parameter and $k=k(\zeta)$ satisfies a third-order ordinary differential equation, found by direct substitution of~\eqref{eq:selfsimilar} into~\eqref{eq:PDEstatement}:
\begin{align}
\nonumber
&\left(\alpha^2-4\zeta\right) k'(\zeta)^2-k'(\zeta) \left(\alpha^2+2\zeta \left(\zeta k'''(\zeta)+k''(\zeta)\right)-\zeta\right)
\\
\label{eq:keq}
&\quad+4 \zeta k'(\zeta)^3-\frac{1}{2} k(\zeta) \left(1-2 k'(\zeta)\right)^2+\zeta\left(\zeta k'''(\zeta)+\zeta k''(\zeta)^2+k''(\zeta)\right)=0.
\end{align}
This equation admits a reduction to the Painlev\'e\,V equation~\eqref{eq:PVTW}, cf.~Remark~\ref{rem:new}.

\subsubsection{Relation to Painlev\'e V and to a system of coupled Painlev\'e V equations}

Multiply~\eqref{eq:Schrodingerstatement} by~${v_t}^2$ and use~\eqref{eq:PDEstatement} to substitute the term $2v_x{v_t}^2$:
\be
{v_t}^2(f_{xx}+\lambda f)-\bigl( t {v_t}^2-\tfrac 14{v_{xt}}^2+\tfrac 12v_tv_{xxt}+\tfrac {\alpha^2}4\bigr)f=0.
\ee
Next, substitute $v_t=\frac 12\bigl(x+\smallint_{0}^{+\infty}f^2\mathrm{d}\sigma\bigr)$ by~\eqref{eq:nonlocal} to get an identity involving $f=f(x;\lambda,t)$ only:
\begin{align}
\nonumber
&\biggl(x+\smallint_{0}^{+\infty}f^2\mathrm{d}\sigma\biggr)^2(f_{xx}+(\lambda-t) f)
+\biggl(\frac 12+\smallint_{0}^{+\infty}ff_x\mathrm{d}\sigma\biggr)^2f
\\
&\qquad\qquad-\biggl(x+\smallint_0^{+\infty}f^2\mathrm{d}\sigma\biggr)\biggl(\smallint_0^{+\infty}\bigl(ff_{xx}+{f_x}^2\bigr)\mathrm{d}\sigma\biggr)f
=\alpha^2f.
\end{align}
Here we denote $\smallint_0^{+\infty}f^2\mathrm{d}\sigma=\smallint_0^{+\infty}f( x;\lambda,t)^2\,\mathrm{d}\sigma(\lambda)$ in the interest of lighter notation, and similarly for the other integrals in this equation.
Finally, let us introduce
\be
g(x;\lambda)=x^{-\frac 14}f(x^{\frac 12};\lambda,0)
\ee
in order to obtain the following consequence of Theorem~\ref{thm}, where we set $t=0$ without loss of generality.

\begin{corollary}
For all $x>0$ such that $Q_\sigma(\sqrt x,0)\not=0$ (in particular, if $0\leq\sigma\leq 1$, for all $x>0$), we have
\begin{align}
\nonumber
Q_\sigma(\sqrt x,0)&=1+\sum_{n=1}^{+\infty}\frac{(-1)^n}{n!}\int_{(0,+\infty)^n}\det\bigl(K_{\mathsf{Be}}(\lambda_i,\lambda_j)\bigr)_{i,j=1}^n\prod_{i=1}^n\sigma(x^{-1}\lambda_i)\mathrm{d} \lambda_i
\\
\label{eq:corollary}
&=\exp\biggl(\frac 14\int_0^x\log\biggl(\frac xy\biggr) \biggl(\int_0 ^{+\infty}\lambda\,g(y;\lambda)^2\,\mathrm{d}\sigma(\lambda)\biggr)\mathrm{d} y\biggr)
\end{align}
where $g=g(x;\lambda)$ satisfies
\begin{align}
\label{eq:idPV}
&-x g\bigl(1+\mathcal I\bigr[g^2\bigr]\bigr)\mathcal I\bigr[ (x g g')'\bigr]+x\bigl(1+\mathcal I\bigr[g^2\bigr]\bigr)^2\biggl(( x g')'+\frac \lambda 4 g\biggr)+ x^2g\bigl(\mathcal I\bigr[ gg'\bigr]\bigr)^2=\frac{\alpha^2}4 g,
\\
&g(x;\lambda)\sim\mathrm{J}_\alpha(\sqrt{x\lambda}),\quad x\to 0_+,
\end{align}
where $'=\partial_x$ and $\mathcal I [h] = \int_{0}^{+\infty}h(x,\lambda)\mathrm{d}\sigma(\lambda)$.
\end{corollary}
One could call~\eqref{eq:idPV} ``integro-differential (or, nonlocal) Painlev\'e\,V equation'', inspired by the analogy with the ``integro-differential Painlev\'e\,II equation'' of Amir--Corwin--Quastel~\cite{ACQ,CCR}.
The latter is connected with the Airy process and the (cylindrical) Korteweg--de\,Vries equation, in the same way as~\eqref{eq:idPV} is connected with the Bessel process and to the nonlinear partial differential equation~\eqref{eq:PDEstatement}, as well as to the narrow-wedge solution to \textit{Kardar--Parisi--Zhang equation}, cf.~\cite{ACQ,BG}.

Since~\eqref{eq:idPV} involves derivatives in~$x$ and integrals in $\lambda$, it is not properly an integro-differential equation.
The equation might be better interpreted, following~\cite{CD}, as an infinite system of coupled Painlev\'e\,V equations.
Indeed, let $\sigma=(1-s) 1_{(-\infty,1]}$ for some $0\leq s\leq 1$, such that $\mathcal I[h]=(s-1) h(x;1)$ and thus~\eqref{eq:corollary} reduces to the Tracy--Widom formula~\cite{TW}
\be
1+\sum_{n=1}^{+\infty}\frac{(s-1)^n}{n!}\int_{(0,x)^n}\det\bigl(K_{\mathsf{Be}}(\lambda_i,\lambda_j)\bigr)_{i,j=1}^n\prod_{i=1}^n\mathrm{d} \lambda_i
= \exp\biggl(-\frac 14\int_0^x\log\biggl(\frac xy\biggr) \,q(y)^2 \,\mathrm{d} y\biggr)
\ee
and~\eqref{eq:idPV} reduces to
\begin{align}
\label{eq:PVTW}
&x q\bigl(1-q^2\bigr)(x qq')'+x\bigl(1-q^2\bigr)^2\biggl(( x q')'+\frac 1 4 q\biggr)+x^2q^3(q')^2=\frac{\alpha^2}4 q,\\
\label{eq:PVTWbc}
&q(x)\sim \sqrt{1-s}\, \mathrm{J}_\alpha(\sqrt{x}),\qquad x\to 0_+,
\end{align}
in terms of $q(x)=\sqrt{1-s}\,g(x;1)$.
This agrees with~\cite[equations~(1.16), (1.17), (1.19)]{TW}.
It is worth recalling ~from loc. cit. that this equation is equivalent to a special case of the Painlev\'e\,V equation through a rational transformation.

\begin{remark}
\label{rem:new}
Let $\mathcal T(x)=0$ be the equation~\eqref{eq:PVTW} and let $\mathcal S(\zeta)=0$ be the equation~\eqref{eq:keq} under the substitution $k(\zeta)=\frac 12\int_0^\zeta q(\eta)^2\mathrm{d}\eta$.
We claim that if $q(x)$ solves $\mathcal T(x)=0$ and satisfies the boundary condition~\eqref{eq:PVTWbc}, then $q(x)$ also satisfies $\mathcal S(x)=0$.
To prove this, let us first note that, in general,
\begin{align}
\frac{\mathcal S(x)-q(x)\mathcal T(x)}{\bigl(1-q(x)^2\bigr)^2}&=\frac 14\biggl(\frac{\bigl(\alpha^2-x\bigr)q(x)^2+xq(x)^4-4x^2q'(x)^2}{1-q(x)^2}+\int_0^xq(y)^2\mathrm d y\biggr) ,
\\
\partial_x\biggl(\frac{\mathcal S(x)-q(x)\mathcal T(x)}{\bigl(1-q(x)^2\bigr)^2}\biggr)&=\frac{2q'(x)\mathcal T(x)}{\bigl(1-q(x)^2\bigr)^2} .
\end{align}
By the second relation, if $\mathcal T(x)=0$ then $\mathcal S(x) = b\bigl(1-q(x)^2\bigr)^2$ for some constant~$b$, and we aim at showing that $b=0$.
Using the first relation, we obtain
\be
\label{eq:lastasymprel}
4b=\frac{(\alpha^2-x)q(x)^2+xq(x)^4-4x^2q'(x)^2}{1-q(x)^2}+\int_0^xq(y)^2\mathrm d y.
\ee
\begin{itemize}
\item
Assume $\alpha>0$ and (in order to obtain a contradiction) $b\not=0$.
Using~\eqref{eq:PVTWbc}, the relation~\eqref{eq:lastasymprel} implies
\be
4b\sim \alpha^2(1-s)\mathrm J_\alpha(\sqrt x)^2-4x^2q'(x)^2,\qquad x\to 0_+ .
\ee
Therefore, $q'(x)\sim\pm \sqrt{-b}\, x^{-1}$ as $x\to 0_+$ and integrating between $x/2$ and $x$ this relation in the limit $x\to 0_+$ we obtain a contradiction with~\eqref{eq:PVTWbc}.
Hence, $b=0$.
\item
Assume $-1<\alpha\leq 0$.
Using~\eqref{eq:PVTWbc}, the relation~\eqref{eq:lastasymprel} implies
\be
4b\sim -\alpha^2+\frac{4x^2q'(x)^2}{q(x)^2},\qquad x\to 0_+ .
\ee
Therefore $q'(x)\sim \pm \sqrt{b^2+\frac{\alpha^2}4}\,\sqrt{1-s}\,\mathrm J_\alpha(\sqrt x)\, x^{-1}$ as $x\to 0_+$ and integrating between $x/2$ and $x$ this relation in the limit $x\to 0_+$ we obtain a contradiction with~\eqref{eq:PVTWbc} unless $b=0$.
\end{itemize}
\end{remark}

More generally, let $k\geq 1$, $-\infty=r_0<r_1<r_2<\dots<r_k$, $0\leq s_1,\dots,s_k\leq 1$ with $s_i\not=s_{i+1}$, and let $\sigma=\sum_{j=1}^k(1-s_j) 1_{(r_{j-1},r_j]}$. (The probabilistic relevance of this case is recalled in Section~\ref{sec:GF} below.)
Then, $\mathcal I[h]=\sum_{j=1}^k(s_{j}-s_{j+1})h(x;r_j)$ (with $s_{k+1}=1$) and thus the corollary above reduces to results of Charlier and Doeraene~\cite[Theorem~1.1]{CD}, namely,~\eqref{eq:corollary} reduces to
\be
Q_\sigma(\sqrt x,0)=\prod_{j=1}^k\exp\biggl(-\frac{r_j}4\int_0^x\log\biggl(\frac xy\biggr) \,q_j(y)^2 \,\mathrm{d} y\biggr)
\ee
and~\eqref{eq:idPV} reduces to a ``system of coupled Painlev\'e\,V equations''
\begin{align}
\label{eq:system}
&x q_j\biggl(1-\sum_{l=1}^kq_l^2\biggr)\sum_{l=1}^k(x q_lq_l')'+x\biggl(1-\sum_{l=1}^kq_l^2\biggr)^2\biggl(( x q_j')'+\frac{r_j} 4 q_j\biggr) +x^2q_j\biggl(\sum_{l=1}^k q_jq_j'\biggr)^2=\frac{\alpha^2}4 q_j,
\\
&q_j(x)\sim\sqrt{s_{j+1}-s_j}\,\mathrm{J}_\alpha(\sqrt{x r_j}),\quad x\to 0_+,
\end{align}
(for all $j=1,\dots,k$) in terms of $q_j(x)=\sqrt{s_{j+1}-s_j}\,g(x;r_j)$.

Therefore, one may interpret~\eqref{eq:idPV} as a continuum limit as $k\to+\infty$ of the system~\eqref{eq:system}.

\subsection{A boundary value problem}

As our final result, we give an application to the partial differential equation~\eqref{eq:PDEstatement} of a scheme devised by Claeys and Tarricone (cf.~\cite[Remark~2.13]{CT}) to solve certain boundary value problems for partial differential equations whose solutions are expressed by Fredholm determinants.

For this purpose, we require more stringent conditions on $\sigma$, namely, smoothness and a strong decay at $+\infty$.
Let $\mathscr S_+$ be the space of infinitely differentiable functions $\sigma:\mathbb R\to\mathbb R$ such that
\be
r^a\frac{\mathrm{d}^b}{\mathrm{d} r^b}\sigma(r)\to 0
\ee
as $r\to+\infty$, for all integers $a,b\geq 0$.

\begin{theorem}
\label{thm2}
If $\sigma\in\mathscr S_+$ and $t\in\mathbb R$ is fixed,
\be
\label{eq:boundaryvalue}
v(x,t)-\frac{xt}2+\frac{4\alpha^2-1}{8x} = x^{2\alpha+1}\frac{2^{-2\alpha-1}}{\Gamma(\alpha+1)^2}\int_t^{+\infty}(\lambda-t)^\alpha\sigma(\lambda)\mathrm{d}\lambda+o(x^{2\alpha+1}),\quad x\to 0_+.
\ee
\end{theorem}

The proof is in Section~\ref{sec:proof2}.
We remind that when $t\in\mathbb R$ is fixed, $Q_\sigma(x,t)\not=0$ for $x>0$ sufficiently small (cf.~Theorem~\ref{thm}).

Remarkably, given $v_0\in\mathscr S_+$, we can construct \textit{explicitly} a solution $v(x,t)$ to~\eqref{eq:PDEstatement} such that
\be
\label{eq:boundaryvaluepb}
v(x,t)-\frac{xt}2+\frac{4\alpha^2-1}{8x} = x^{2\alpha+1}v_0(t)+o(x^{2\alpha+1}),\quad x\to 0_+.
\ee
To this end, for $f\in\mathscr S_+$, let us introduce the \textit{Riemann--Liouville integral}
\be 
\label{eq:W}
(\mathcal M_{\mu}f)(t) = \frac 1{\Gamma(\mu)}\int_t^{+\infty}(\lambda-t)^{\mu-1}f(\lambda)\mathrm{d}\lambda.
\ee
The map $\mu\mapsto(\mathcal M_\mu f)(t)$ is entire in $\mu$ for all $t\in\mathbb R$.
Indeed,~\eqref{eq:W} can be analytically continued to $\mu\in\mathbb C$ by integrating by parts, namely,
\be
\label{eq:analyticcontinuation}
(\mathcal M_{\mu}f)(t) = \frac{(-1)^k}{\Gamma(\mu+k)}\int_t^{+\infty}(\lambda-t)^{\mu+k-1}f^{(k)}(\lambda)\mathrm{d}\lambda,
\ee
which is convergent and independent of the integer $k$ as soon as $\Re\mu>-k$.
On the other hand, the map $t\mapsto(\mathcal M_\mu f)(t)$ is in~$\mathscr S_+$ for all $\mu\in\mathbb C$.
Moreover, we have $\mathcal M_\mu\circ\mathcal M_{\nu}=\mathcal M_{\mu+\nu}$, which implies the inversion formula $\mathcal M_\mu(\mathcal M_{-\mu} f)=f$.
Therefore, the function sought is $v(x,t)$ defined in~\eqref{eq:defvstatement} with
\begin{align}
\nonumber
\sigma(\lambda)&=2^{2\alpha+1}\Gamma(\alpha+1)\bigl(\mathcal M_{-\alpha-1}v_0\bigr)(\lambda)
\\
\label{eq:sigmafromv0}
&=2^{2\alpha+1}(-1)^k\frac{\Gamma(\alpha+1)}{\Gamma(k-\alpha-1)}\int_{\lambda}^{+\infty}(t-\lambda)^{k-\alpha-2}v_0^{(k)}(t)\mathrm{d} t,
\end{align}
for an integer $k$ sufficiently large (namely, $k>\alpha+1$), cf.~\eqref{eq:analyticcontinuation}.

\section{Probability and operator preliminaries}\label{sec:probability}

\subsection{The Bessel kernel operator}
In general, it turns out to be convenient to associate to a determinantal point process on a space~$\mathfrak X$ an integral kernel operator on~${L}^2(\mathfrak X)$ acting via the correlation kernel~\cite[Section~1]{S}.
In the case of the Bessel process, this is the operator $\mathcal K_{\mathsf{Be}}$, acting on~${L}^2(0,+\infty)$ by
\be
(\mathcal K_{\mathsf{Be}}\psi)(\lambda)=\int_0^{+\infty}K_{\mathsf{Be}}(\lambda,\mu)\psi(\mu)\,\mathrm{d} \mu,\qquad \psi\in{L}^2(0,+\infty),
\ee
with $K_{\mathsf{Be}}(\lambda,\mu)$ as in~\eqref{eq:Besselkernel}.
It is well-known that the Bessel kernel admits the equivalent expression
\be
\label{eq:conv}
K_{\mathsf{Be}}(\lambda,\mu) = \frac 14\int_0^1\mathrm{J}_\alpha(\sqrt{\rho\lambda})\mathrm{J}_\alpha(\sqrt{\rho\mu})\mathrm{d}\rho.
\ee

\begin{remark}
To prove~\eqref{eq:conv}, consider the Bessel differential equation in Sturm--Liouville form:
\be
\biggl(\frac \lambda 4-\frac{\alpha^2}{4\rho}\biggr)\mathrm{J}_\alpha(\sqrt{\rho\lambda}) = \partial_\rho\biggl(\rho\,\partial_\rho\,\mathrm{J}_\alpha(\sqrt{\rho\lambda})\biggr) .
\ee
Multiply it by $\mathrm{J}_\alpha(\sqrt{\rho\mu})$ and then subtract from the resulting identity the same identity with $\lambda,\mu$ swapped.
Finally,~\eqref{eq:conv} follows by an integration by parts.
\end{remark}

Such formula for the Bessel kernel allows us to write the following identity of operators
\be
\label{eq:Kisproj}
\mathcal K_{\mathsf{Be}}=\mathcal J_\alpha 1_{(0,1)}\mathcal J_\alpha
\ee
where $1_{(0,1)}$ denotes multiplication by the characteristic function of the unit interval, i.e., orthogonal projection onto ${L}^2(0,1)$, and $\mathcal J_\alpha$ is the unitary involution of ${L}^2(0,+\infty)$ defined by
\be
\label{eq:Besseltransform}
(\mathcal J_\alpha\psi)(\lambda)=\frac 12\int_0^{+\infty}\mathrm{J}_\alpha(\sqrt{\lambda\mu})\psi(\mu)\mathrm{d}\mu,\qquad\psi\in{L}^2(0,+\infty).
\ee

\begin{remark}
Strictly speaking,~\eqref{eq:Besseltransform} holds true for $\psi$ smooth with compact support, but it can be extended to $\psi$ in~${L}^2$ by an argument parallel to the usual definition of the Fourier transform on ${L}^2(\mathbb R)$. 
Note also that~\eqref{eq:Besseltransform} is, up to passing to the square root variable, the {\em Hankel transform}~\cite[Chapter~9]{Akhiezer}; indeed, the latter is usually defined as $\widehat\psi(k)=\int_0^{+\infty}\mathrm{J}_\alpha(kr)\psi(r)r\mathrm{d} r$, which implies $\widehat\psi(\sqrt k)=\frac 12\int_0^{+\infty}\mathrm{J}_\alpha(\sqrt{kr})\psi(\sqrt r)\mathrm{d} r=\mathcal J_\alpha\bigl(\psi(\sqrt\cdot)\bigr)(k)$.
\end{remark}

Introducing the multiplication operator $\mathcal M_{ x , t }^\sigma$ on ${L}^2(0,+\infty)$ (parametrically depending on $x>0$, $t\in\mathbb R$) by
\be
(\mathcal M_{ x , t }^\sigma\psi)(\lambda)=\sqrt{\sigma( x^{-2}\lambda + t )}\,\psi(\lambda),\qquad\psi\in{L}^2(0,+\infty),
\ee
we have
\be
\label{eq:Qdet}
Q_\sigma( x , t )=\det_{{L}^2(0,+\infty)}(1-\mathcal K_{ x , t }^\sigma),\qquad \mathcal K_{ x , t }^\sigma=\mathcal M_{ x , t }^\sigma\mathcal K_{\mathsf{Be}}\mathcal M_{ x , t }^\sigma.
\ee

Let us denote $\mathfrak J_p$ the ideal of operators $\mathcal A$ on~$L^2(0,+\infty)$ such that the Schatten $p$-norm
\be
\|\mathcal A\|_p=\bigl(\tr |\mathcal A|^p\bigr)^{1/p}
\ee
is finite.
We will need the cases $p=1$ (trace-class operators) and $p=2$ (Hilbert--Schmidt operators) and in particular the following characterizations of~$\mathfrak J_1$ and $\mathfrak J_2$, cf.~\cite{Simon}.
\begin{itemize}
\item $\mathcal A\in\mathfrak J_2$ if and only if
\be
\bigl(\mathcal A \psi\bigr) (u) = \int_0^{+\infty} a(u,v)\psi(v)\mathrm{d} y
\ee
for some $a\in L^ 2\bigl((0,+\infty)\times(0,+\infty)\bigr)$, in which case
\be
\|\mathcal A\|_2=\int_0^{+\infty}\int_0^{+\infty}\bigl|a(u,v)\bigr|^2\mathrm{d} u\mathrm{d} v  .
\ee 
\item $\mathcal A\in\mathfrak J_1$ if and only $\mathcal A=\mathcal B\,\mathcal C$ for some~$\mathcal B,\mathcal C\in\mathfrak J_2$, in which case $\|\mathcal A\|_1\leq \|\mathcal B\|_2\,\|\mathcal C\|_2$.
\end{itemize}

\begin{lemma}\label{lemma:traceclass1}
The following facts hold true.

\textit{(i)} $\mathcal K^\sigma_{ x , t }$ is self-adjoint.

\textit{(ii)} $\mathcal K^\sigma_{ x , t }\in\mathfrak J_1$.

\textit{(iii)}  If $0\leq\sigma\leq 1$, we have $0\leq\mathcal K^\sigma_{ x , t }<1$.

\textit{(iv)} $\|\mathcal K^\sigma_{x,t}\|_1=O(x^{2+2\alpha})$ as $x\to 0_+$.
\end{lemma}
\begin{proof}~

\textit{(i)}
Self-adjointness follows directly from~\eqref{eq:Kisproj}.

\textit{(ii)}
Write $\mathcal K^\sigma_{x,t}=\mathcal H\mathcal H^\dagger$ where $\mathcal H=\mathcal M^\sigma_{x,t}\mathcal J_\alpha 1_{(0,1)}$.
Note that
\begin{align}
\nonumber
\|\mathcal H\|_{2}^2&=\frac 14\int_0^{+\infty}\sigma( x ^{-2}\lambda+ t )\biggl(\int_0^1 \mathrm{J}_\alpha^2(\sqrt{\lambda\mu})\mathrm{d}\mu\biggr)\mathrm{d}\lambda
\\
\label{eq:finitetrace}
&=\frac 14\int_0^{+\infty}\sigma( x ^{-2}\lambda+ t )K_{\mathsf{Be}}(\lambda,\lambda)\mathrm{d}\lambda<+\infty,
\end{align}
where we use $K_{\mathsf{Be}}(\lambda,\lambda)=O(\lambda^{-\frac 12})$ as $\lambda\to+\infty$ and~Assumption~\ref{assumption}.
Hence $\mathcal H,\mathcal H^\dagger\in\mathfrak J_2$ which implies $\mathcal K^\sigma_{ x , t }\in\mathfrak J_1$.

\textit{(iii)}
The weak inequalities $0\leq\mathcal K^\sigma_{ x , t }\leq 1$ follow directly from the assumption $0\leq \sigma\leq 1$ because $\mathcal K_{\mathsf{Be}}$ is a projection operator, cf.~\eqref{eq:Kisproj}. 
Therefore, we only have to show that $1$ is not an eigenvalue of the compact operator $\mathcal K^\sigma_{ x , t }$, which implies the strict inequality $\mathcal K^\sigma_{x,t}<1$.
To this end, assume that $\psi\in{L}^2(0,+\infty)$ satisfies $\psi=\mathcal K^\sigma_{ x , t }\psi$; our goal is to show that $\psi=0$.
Set $\widetilde\psi=\mathcal K_{\mathsf{Be}} \mathcal M_{ x , t }^\sigma\psi$, such that $\psi=\mathcal K^\sigma_{ x , t }\psi$ implies $\sqrt{\sigma( x^{-2}\lambda+ t )}\widetilde\psi(\lambda)=\psi(\lambda)$.
Moreover, since $\mathcal K_{\mathsf{Be}}$ is a projection operator and $0\leq\sigma\leq 1$, the following chain of inequalities
\begin{align}
\nonumber
\|\widetilde\psi(\cdot)\|_{L^2(0,+\infty)}&\leq\|\sqrt{\sigma(x^{-2}\cdot + t )}\psi(\cdot)\|_{L^2(0,+\infty)}
\\
&=\|\sigma( x ^{-2}\cdot+ t )\widetilde\psi(\cdot)\|_{L^2(0,+\infty)}\leq\|\widetilde\psi(\cdot)\|_{L^2(0,+\infty)}
\end{align}
collapses to a chain of equalities, implying that $(1-\sigma( x ^{-2}\lambda+ t ))\widetilde\psi(\lambda)=0$ almost everywhere in~$\lambda\in (0,+\infty)$.
By Assumption~\ref{assumption}, this implies that $\widetilde\psi$ vanishes on some interval of the form $(s,+\infty)$.
On the other hand, $\widetilde\psi$ belongs to the range of $\mathcal K_{\mathsf{Be}}$ which, by~\eqref{eq:Kisproj} and properties of the Bessel functions, consists of functions with an analytic extension to an open neighborhood of the positive real half-line.
Therefore $\widetilde\psi$ must be identically zero, and so does $\psi$.

\textit{(iv)}
By standard properties of Bessel functions, cf.~\cite[\S~10]{DLMF}, $K_{\mathsf{Be}}(\lambda,\lambda)$ is a smooth function of $\lambda>0$, bounded in every compact set of $(0,+\infty)$ and such that $K_\mathsf{Be}(\lambda,\lambda)=O(\lambda^{\alpha})$ as $\lambda\to 0_+$ and $K_{\mathsf{Be}}(\lambda,\lambda)=O(\lambda^{-\frac 12})$ as $\lambda\to+\infty$.
It follows that $K_\mathsf{Be}(\lambda,\lambda)=O(\lambda^\alpha)$ for $0<\lambda\leq 1$ and $K_{\mathsf{Be}}(\lambda,\lambda)=O(\lambda^{-\frac 12})$ for $\lambda\geq 1$.
On the other hand, by Assumption~\ref{assumption} we have, for any fixed $t$, $\sigma(r)=O\bigl(\min(1,(r-t)^{-b})\bigr)$ for all $r>t$ and for an arbitrarily large $b>0$.
Therefore, by~\eqref{eq:finitetrace} we have, assuming $x<1$ and $b$ large enough,
\begin{equation}
\label{eq:last2}
\|\mathcal H\|_{2}^2=O\biggl(\int_0^{x^2}\lambda^\alpha\mathrm d\lambda + \int_{x^2}^1\lambda^\alpha(x^{-2}\lambda)^{-b}\mathrm d\lambda
+\int_1^{+\infty}\lambda^{-\frac 12}(x^{-2}\lambda)^{-b}\mathrm d\lambda\biggr)=O(x^{2+2\alpha})
\end{equation}
and the proof is complete.
\end{proof}

\begin{corollary}
\label{corollary:traceclass}
The right-hand side of~\eqref{eq:Qdet} is a well-defined Fredholm determinant and $\log Q_\sigma(x,t)=O(x^{2+2\alpha})$ as $x\to 0_+$.
Moreover, if $0\leq \sigma\leq 1$ we have $0<Q_\sigma(x,t)\leq 1$ for all $x\in\mathbb R$, $t>0$.
\end{corollary}

\begin{proof}
Recall that $\bigl|\det(1-\mathcal L)-1\bigr|\leq \|\mathcal L\|_1\,\exp\bigl(1+\|\mathcal L\|_1\bigr)$ for any $\mathcal L\in\mathfrak J_1$ (cf.~\cite{Simon}).
\end{proof}

To show that $Q_\sigma(x,t)$ is smooth in $x,t$ we perform a change of variables and write
\be
Q_\sigma(x,t)=\det_{L^2(t,+\infty)}(1-\widetilde{\mathcal K}_{x,t}^\sigma)
\ee
where $\widetilde{\mathcal K}_{x,t}^\sigma$ acts on $L^2(t,+\infty)$ through the kernel 
\be
x^2 \sqrt{\sigma(u)\sigma(v)}K_{\mathsf{Be}}\bigl(x^2(u-t),x^2(v-t)\bigr),\qquad u,v>t.
\ee

\begin{lemma}\label{lemma:traceclass2}
The map $x\mapsto Q_\sigma(x,t)$ is smooth for all $x>0$ and, whenever $Q_\sigma(x,t)\not=0$, we have
\be
\label{eq:Jacobifinal}
\partial_x\log Q_\sigma(x,t)=-\tr_{L^2(t,+\infty)}\biggl((1-\widetilde{\mathcal K}_{x,t}^\sigma)^{-1}\,\partial_x\widetilde{\mathcal K}_{x,t}^\sigma\biggr).
\ee 
\end{lemma}

\begin{proof}
Similarly as before, we may write $\widetilde{\mathcal K}^\sigma_{x,t}=\widetilde{\mathcal H}\widetilde{\mathcal H}^\dagger$ where $\widetilde{\mathcal H}:L^2(0,1)\to L^2(t,+\infty)$ acts via the kernel $H(u,\rho)=x\sqrt{\sigma(u)}\mathrm J_\alpha\bigl(x\sqrt{\rho(u-t)}\bigr)$.
The latter, as a function of $x$, is infinitely differentiable, with $x$-derivatives which are still in $L^2\bigl((t,+\infty)\times(0,1),\mathrm d u\,\mathrm d \rho\bigr)$ and so we can conclude (thanks to the fast decay of $\sigma$ at $+\infty$) that $\widetilde{\mathcal H}$ is infinitely differentiable in $x$ with respect to the Hilbert--Schmidt norm.
Hence, $\widetilde{\mathcal K}_{x,t}^\sigma$ is infinitely differentiable in $x$ with respect to the trace-class norm, and so $Q_\sigma(x,t)$ is smooth in~$x$.
Finally,~\eqref{eq:Jacobifinal} is an application of the well-known Jacobi's formula, cf.~\cite{Widom}.
\end{proof}

Finally, let us prove a lemma which will be useful later on.

\begin{lemma}
\label{lemma:Bessel}
Let $j_\alpha$ be the smallest positive zero of the Bessel function $\mathrm{J}_\alpha$.
We have
\be
\frac{K_{\mathsf{Be}}\bigl(\xi^2,\mu\bigr)}{\mathrm{J}_\alpha(\xi)}=\begin{cases}
O\bigl(\mu^{\frac \alpha 2}\bigr),&0<\mu\leq 1, \\
O\bigl(\mu^{-\frac 14}\bigr),&\mu\geq 1,
\end{cases}
\ee
uniformly for $0<\xi\leq\xi_0$ for all $\xi_0<\min(j_\alpha,j_{\alpha+1})$.
\end{lemma}
\begin{proof}
Let us recall that $\mathrm{J}_\alpha(z)$ is positive when $0<z<j_\alpha$.
Moreover, $\partial_z(\mathrm{J}_\alpha(z)z^{-\alpha})=-z^{-\alpha}\mathrm{J}_{\alpha+1}(z)$ such that the map $z\mapsto \mathrm{J}_\alpha(z)z^{-\alpha}$ is positive and monotonically decreasing for $0<z<\min(j_\alpha,j_{\alpha+1})$.
Hence, the map $\rho\mapsto\mathrm{J}_\alpha(\sqrt{\rho}\xi)\rho^{-\frac \alpha 2}$ is positive and monotonically decreasing for $0<\rho<1$ for all $0<\xi<\min(j_\alpha,j_{\alpha+1})$, implying that
\be
\mathrm{J}_\alpha(\sqrt{\rho}\xi)\rho^{-\frac \alpha 2}\leq\lim_{\rho\to 0_+}\mathrm{J}_\alpha(\sqrt{\rho}\xi)\rho^{-\frac \alpha 2}=\frac{\xi^{\alpha}}{2^{\alpha}\Gamma(1+\alpha)}
\ee
for all $0<\rho<1$, $0<\xi<\min(j_\alpha,j_{\alpha+1})$.
Therefore, if $0<\xi<\xi_0<\min(j_\alpha,j_{\alpha+1})$ and if $0<\rho<1$,
\be
\label{eq:BesselIneq}
\frac{\mathrm{J}_\alpha(\sqrt{\rho}\xi)}{\mathrm{J}_\alpha(\xi)}\rho^{-\frac \alpha 2}\leq\frac{\xi_0^{\alpha}}{2^{\alpha}\Gamma(1+\alpha)\mathrm{J}_\alpha(\xi_0)}.
\ee
Hence, by~\eqref{eq:conv} we can bound
\be
\frac{K_{\mathsf{Be}}\bigl(\xi^2,\mu\bigr)}{\mathrm{J}_\alpha(\xi)}=\frac 14\int_0^1\mathrm{J}_\alpha(\sqrt{\rho\mu})\frac{\mathrm{J}_\alpha(\xi\sqrt{\rho})}{\mathrm{J}_\alpha(\xi)}\mathrm d\rho=O\biggl( \int_0^1|\mathrm{J}_\alpha(\sqrt{\rho\mu})|\rho^{\frac \alpha 2}\mathrm d\rho\biggr)
\ee
uniformly for $0<\xi<\xi_0$ for all $\xi_0<\min(j_\alpha,j_{\alpha+1})$.
Next, since $\mathrm{J}_\alpha(z)=O(z^\alpha)$ for $0<z\leq 1$ and $\mathrm{J}_\alpha(z)=O(z^{-\frac 12})$ for $z\geq 1$, the following estimates complete the proof.
\begin{itemize}
\item When $0<\mu\leq 1$, $\int_0^1|\mathrm{J}_\alpha(\sqrt{\rho\mu})|\rho^{\frac \alpha 2}\mathrm d\rho=O\bigl(\mu^{\frac\alpha 2}\int_0^1\rho^{\alpha}\mathrm d\rho\bigr)=O(\mu^{\frac \alpha 2})$.
\item When $\mu\geq 1$, $\int_0^1|\mathrm{J}_\alpha(\sqrt{\rho\mu})|\rho^{\frac \alpha 2}\mathrm d\rho=O\bigl(\mu^{\frac\alpha 2}\int_0^{\frac 1\mu}\rho^\alpha\mathrm d\rho+\mu^{-\frac 14}\int_{\frac 1\mu}^1\rho^{-\frac 14+\frac \alpha 2}\mathrm d\rho\bigr)=O(\mu^{-\frac 14})$.
\end{itemize}
\end{proof}

\subsection{Multiplicative statistics of the Bessel process}\label{sec:GF}

It follows from basic properties of determinantal point process~\cite[Theorem~4]{S} that the Bessel process has (almost surely) infinitely many particles and a smallest particle.
Denoting $\lambda_1<\lambda_2<\cdots$ these particles, we have, again by standard properties of determinantal point processes~\cite[eq.~2.4]{B},
\be
\label{eq:multiplicative}
Q_\sigma( x , t )=\mathbb E\prod_{i\geq 1}\bigl(1-\sigma(x^{-2} \lambda_i + t )\bigr).
\ee
Therefore, the quantity of interest in this paper is naturally a \textit{multiplicative statistics} of the Bessel process.

In particular, when $\sigma=\sum_{j=1}^k(1-s_j) 1_{[r_{j-1},r_j)}$, for some $k\geq 1$, $-\infty=r_0<r_1<r_2<\dots<r_k$, and $0\leq s_1,\dots,s_k\leq 1$ with $s_i\not=s_{i+1}$, $Q_\sigma( x , t )$ is called \textit{generating function} of the process~\cite{CD}, because its Taylor coefficients at $s_j=0$ encode natural counting statistics of the Bessel process.
Namely, let $\#_j$ be the random variable counting the number of particles $\lambda_i$ satisfying $x^2(r_{j-1}-t)<\lambda_i<x^2(r_j-t)$, then
\be
Q_\sigma(x,t)=\mathbb E\prod_{j=1}^ks_j^{\#_j} = 
\sum_{m_1,\dots,m_k\geq 0}\mathbb P\bigl(\#_1=m_1,\dots,\#_k=m_k\bigr)\, s_1^{m_1}\cdots s_k^{m_k}.
\ee
Moreover, when $\sigma=1_{(-\infty,r]}$, $Q_\sigma(x,t)$ gives the smallest particle distribution:
\be
Q_\sigma(x,t)=\mathbb P\bigl(\lambda_1>x^2(r-t )\bigr).
\ee

\subsection{Full gap probability of a thinned Bessel process}

In this paragraph, we assume $0\leq \sigma\leq 1$.
Let the \textit{thinned} Bessel process be the random point process obtained from the Bessel process by removing each particle independently with probability $1-\sigma( x^{-2}\lambda + t )$ depending on the location $\lambda\in\mathbb R$ of the particle (as  well as on the usual parameters $x>0$ and $t\in\mathbb R$).
The thinned Bessel process is still determinantal (a fact which holds true in general for thinning of determinantal point processes~\cite{LMR}) with correlation kernel
\be
\sqrt{\sigma(x^{-2}\lambda+t)}\,K_{\mathsf{Be}}(\lambda,\mu)\,\sqrt{\sigma(x^{-2}\mu+t)}.
\ee
Therefore, by using again the general identity for multiplicative statistics of determinantal point processes~\cite[eq.~2.4]{B}, we infer that $Q_\sigma( x , t )$ is the probability that the thinned Bessel process contains no particle at all (a \textit{full gap probability}).

\subsection{Positive-temperature deformation of the Bessel process}\label{sec:PT}

Recalling the factorization $\mathcal K^\sigma_{ x , t }=\mathcal H\mathcal H^\dagger$ introduced in the proof of Lemma~\ref{lemma:traceclass1}, we can write
\be
Q_\sigma( x , t )=\det_{L^2(0,+\infty)}(1-\mathcal H\mathcal H^\dagger)=\det_{L^2(0,+\infty)}(1-\mathcal H^\dagger\mathcal H) = \det_{L^2(0,+\infty)} (1-1_{(0,1)}{\mathcal A}_{ x , t }^\sigma 1_{(0,1)})
\ee
where ${\mathcal A}_{ x , t }^\sigma=\mathcal J_\alpha(\mathcal M_{ x , t }^\sigma)^2\mathcal J_\alpha$, which has integral kernel
\be
A_{ x , t }^\sigma(\lambda,\mu) = \frac 14\int_0^{+\infty}\sigma( x ^{-2}\rho + t )\mathrm{J}_\alpha(\sqrt{\rho\lambda})\mathrm{J}_\alpha(\sqrt{\rho\mu})\mathrm d\rho.
\ee
Setting $ t=0$ and a performing a change of variable, we obtain
\be
Q_\sigma( x ,0)=\det_{L^2(0,+\infty)}(1-1_{(0, x ^{2})}\widetilde{\mathcal K}_\sigma 1_{(0, x ^{2})})
\ee
where $\widetilde{\mathcal K}_\sigma$ has integral kernel
\be
\label{eq:positivetempBessel}
\widetilde{K}_\sigma(\lambda,\mu) =\frac 14\int_0^{+\infty}\sigma(\rho)\mathrm{J}_\alpha(\sqrt{\rho\lambda})\mathrm{J}_\alpha(\sqrt{\rho\mu})\mathrm d\rho,
\ee
which can be regarded as a deformation of the Bessel kernel: when $\sigma=1_{(-\infty,1]}$, we recover $K_{\mathsf{Be}}(\lambda,\mu)=\widetilde K_{1_{(0,1)}}(\lambda,\mu)$, cf.~\eqref{eq:conv}.

By the same general facts about determinantal point processes used above, it is not hard to prove that the kernel $\widetilde K_\sigma$ defines a determinantal point process with almost surely a smallest particle $\lambda_{\mathrm{min}}$ and so $Q_\sigma( x ,0)=\mathbb P(\lambda_{\mathrm{min}}> x ^{2})$ gives the smallest particle distribution.

Kernels of the form~\eqref{eq:positivetempBessel} have recently attracted some interest, cf.~\cite{BO,LLDMS}.

\section{Riemann--Hilbert problems}\label{sec:RH}

\subsection{Integrable operators}

The operator $\mathcal K^\sigma_{ x , t }$, see~\eqref{eq:Qdet}, is \textit{2-integrable}, in the sense of Its--Izergin--Korepin--Slavnov~\cite{IIKS}.
Namely, it acts on ${L}^2(0,+\infty)$ through a kernel of the form
\be
\label{eq:integrable}
\frac{\f(\lambda)^\top \, \g(\mu)}{\lambda-\mu}=\frac{f_1(\lambda)g_1(\mu)+f_2(\lambda)g_2(\mu)}{\lambda-\mu},\qquad
\mathbf f=\begin{pmatrix}
f_1 \\ f_2
\end{pmatrix},\quad
\mathbf g=\begin{pmatrix}
g_1 \\ g_2
\end{pmatrix}.
\ee
In the case of present interest, the vectors $\f(\lambda)$ and $\g(\lambda)$ are given by
\be
\label{eq:fg}
\f(\lambda)=\sqrt{\frac{\sigma(x^{-2}\lambda+t)}{2}}\widehat \f(\lambda),\qquad
\g(\lambda)=\sqrt{\frac{\sigma(x^{-2}\lambda+t)}{2}}\widehat \g(\lambda),
\ee
with
\be
\label{eq:fghat}
\widehat \f(\lambda)=C\begin{pmatrix}\mathrm{J}_\alpha(\sqrt \lambda) \\ \sqrt \lambda\,\mathrm{J}'_\alpha(\sqrt \lambda)\end{pmatrix},\qquad
\widehat \g(\lambda)=C^{-\top}\begin{pmatrix} \sqrt \lambda\,\mathrm{J}_\alpha'(\sqrt \lambda) \\ -\mathrm{J}_\alpha(\sqrt \lambda)\end{pmatrix}.
\ee
Here, $C$ could be any $2\times 2$ invertible matrix independent of~$\lambda$, which we immediately fix for later convenience (cf. Remark~\ref{remark:reasonforC}) as
\be
\label{eq:C}
C=\begin{pmatrix}
1 & 0 \\ d_\alpha & 1
\end{pmatrix},\qquad d_\alpha=\frac{4\alpha^2+3}{8}.
\ee

\begin{remark}~
\begin{itemize}
\item The kernel~\eqref{eq:integrable} is continuous along the diagonal~$\lambda=\mu$.
\item At points where $\sigma<0$, the square root can be chosen arbitrarily.
\end{itemize}
\end{remark}

We now summarize the main points of the well-known Its--Izergin--Korepin--Slavnov theory of integrable operators, adapted to the specific case at hand.
For further information we refer to~\cite{DIZ,HI,IIKS} or~\cite[Chapter~5]{BDS}.

\begin{remark}
When $\alpha\geq 0$, we have $\f,\g\in L^\infty(0,+\infty)$ which is the usual setting for Its--Izergin--Korepin--Slavnov theory of integrable operators, cf.~\cite[Definition~5.18]{BDS}.
Some minor clarifications are needed when $-1<\alpha<0$, in which case $\f,\g\not\in L^\infty(0,+\infty)$, therefore we include the following discussion and Proposition~\ref{prop:new}.
\end{remark}

\medskip

We first recall a few facts about Cauchy operators.

Let us denote, for $h\in L^p(0,+\infty)$ for some $1<p<\infty$,
\be
\mathcal C h (z) = \frac 1{2\pi\mathrm i}\int_0^{+\infty}\frac{h(\lambda)}{\lambda-z}\mathrm d \lambda
\ee
and the boundary values
\be
\mathcal C_\pm h (\lambda) = \lim_{\epsilon\to 0_+} \mathcal C h (\lambda\pm\mathrm i\epsilon)
\ee
which are known to exist for almost all $\lambda>0$ (and they exist more generally as non-tangential limits as the argument of $\mathcal C h$ approaches the positive real line from above or below).
Moreover $\mathcal C_\pm$ are bounded operators on $L^p(0,+\infty)$ for any $1<p<\infty$ and $\mathcal C_\pm h(\lambda)$ is continuous at all points $\lambda$ where $h(\lambda)$ is smooth.
We will also need the following two statements whose proof can be extracted from~\cite[Chapter~4]{Musk}.
\begin{itemize}
\item For any jump discontinuity $r>0$ of $h$ we have
\be
\label{cauchylog}
\mathcal C h(z)=O(\log|z-r|)
\ee
as $z\to r$ away from the real axis.
\item If $h(\lambda)$ is smooth in a right neighborhood of $0$ and $h(\lambda)=O(\lambda^\beta)$ as $\lambda\to 0_+$ for some $\beta>-1$, we have 
\be 
\label{cauchyendpoint}
\mathcal C h(z)=
\begin{cases}
O(\log|z|) &\mbox{if }\beta\geq 0,\\
O(|z|^\beta) &\mbox{if }-1<\beta<0,
\end{cases}
\ee
as $z\to 0$ away from the positive real axis.
\end{itemize}
\medskip

Let us introduce, whenever $Q_\sigma(x,t)\not=0$,
\be
\label{eq:defF}
\mathbf F=\begin{pmatrix}F_1 \\ F_2 
\end{pmatrix}=\begin{pmatrix}(1-\mathcal K_{x,t}^\sigma)^{-1}f_1 \\ 
(1-\mathcal K_{x,t}^\sigma)^{-1}f_2 
\end{pmatrix}.
\ee

Recall from the introduction that we denote $r_1,\dots,r_k$ the discontinuity points of $\sigma(\cdot)$.
Let $\widehat r_i=x^2(r_i-t)$ be the discontinuity points of $\sigma(x^{-2}\cdot +t)$.

\begin{lemma}
\label{lemma:F}
For $j=1,2$, $F_j=F_j(\lambda)$ is smooth for $\lambda\not=\widehat r_i$, has jump discontinuities at $\lambda=\widehat r_i$, and is $O\bigl(\lambda^{\frac \alpha 2}\bigr)$ as $\lambda\to 0_+$.
\end{lemma}
\begin{proof}
The defining relation~\eqref{eq:defF} implies, for $j=1,2$,
\begin{align}
\nonumber
F_j(\lambda)&=\int_0^{+\infty} K^\sigma_{x,t}(\lambda,\mu)F_j(\mu)\mathrm d\mu + f_j(\lambda)
\\
&=\sqrt{\sigma(x^{-2}\lambda+t)}\int_0^{+\infty} K_{\mathsf{Be}}(\lambda,\mu)\sqrt{\sigma(x^{-2}\mu+t)}F_j(\mu)\mathrm d\mu + f_j(\lambda)
\end{align}
and the regularity properties in $\lambda$ are clear.
Moreover, from this relation we also get the estimate
\be
|F_i(\lambda)|\leq |f_i(\lambda)|+\|\sqrt\sigma\|_{L^\infty(t,+\infty)}\,\|F_i\|_{L^2(0,+\infty)}\,\|K_{\mathsf{Be}}(\lambda,\cdot)\sqrt{\sigma(x^2\cdot+t)}\|_{L^2(0,+\infty)}.
\ee
Standard asymptotic properties of Bessel functions imply that $f_i(\lambda)=O(\lambda^{\frac \alpha 2})$ as $\lambda\to 0_+$.
Moreover, by Lemma~\ref{lemma:Bessel} we can bound the last $L^2$-norm, for all $\lambda\in (0,\lambda_0)$ for sufficiently small $\lambda_0$, as follows:
\begin{align}
\nonumber
&\biggl(
\int_0^{+\infty}K_{\mathsf{Be}}(\lambda,\mu)^2\,|\sigma(x^2\mu+t)|\,\mathrm d\mu\biggr)^{\frac 12}
\\
&\qquad\leq M J_\alpha(\sqrt\lambda)\biggl(\|\sigma\|_{L^\infty(t,+\infty)}\int_0^1\mu^\alpha\mathrm d\mu\,+\,\int_{1}^{+\infty}\mu^{-\frac 12}|\sigma(x^2\mu+t)|\,\mathrm d\mu\biggr)^{\frac 12}
\end{align}
for some constant $M$ depending on $\alpha,\lambda_0$ only.
Therefore, also thanks to the fast decay of $\sigma$ at $+\infty$, we have $\|K_{\mathsf{Be}}(\lambda,\cdot)\sqrt{\sigma(x^2\cdot+t)}\|_{L^2(0,+\infty)}=O\bigl(\mathrm{J}_\alpha(\sqrt\lambda)\bigr)=O(\lambda^{\frac\alpha 2})$ as $\lambda\to 0_+$, which completes the proof.
\end{proof}

Next, we introduce, whenever $Q_\sigma(x,t)\not=0$,
\be
\label{eq:Yexplicit}
Y(z)=I-2\pi\mathrm i\, \mathcal C\bigl(\mathbf F\mathbf g^\top\bigr)(z)=I-\int_0^{+\infty}\frac{\mathbf F(\lambda)\,\mathbf g(\lambda)^\top}{\lambda-z}\mathrm d\lambda,\qquad z\in\mathbb C\setminus[0,+\infty).
\ee

\begin{remark}
Hereafter, $I$ denotes the $2\times 2$ identity matrix.
Moreover, we do not indicate explicitly the dependence on $\sigma,x,t$ of various quantities (e.g., $\mathbf f,\mathbf g,Y$) in the interest of lighter notations.
\end{remark}

We shall prove (cf.~Proposition~\ref{prop:new}) that this matrix-valued function of $z$ is uniquely characterized by the following \textit{Riemann--Hilbert problem}.

\begin{RH}\label{RHP:Y}~
\begin{itemize}
\item[(a)] $Y(z)$ is a $2\times 2$ matrix analytically depending on $z\in\mathbb C\setminus[0,+\infty)$. 
\item[(b)] For all $\lambda\in(0,+\infty)\setminus\lbrace \widehat r_1,\dots,\widehat r_k\rbrace$, the boundary values $Y_\pm(\lambda)=\lim_{\epsilon\to 0_+}Y(\lambda\pm\mathrm{i}\epsilon)$ exist, are continuous, and satisfy
\be
\label{eq:Yjump}
Y_+(\lambda) = Y_-(\lambda) J_Y(\lambda),\quad J_Y(\lambda)=I-2\pi\mathrm{i}\,\f(\lambda) \g(\lambda)^\top.
\ee
\item[(c)] As $z\to \widehat r_j$ away from the real axis we have $Y(z)=O\bigl(\log|z-\widehat r_j|\bigr)$, and as $z\to 0$ away from the positive real axis we have
\be
\label{eq:Y0}
Y(z)=\begin{cases}
O\bigl(\log|z|\bigr) &\text{if }\alpha\geq 0,\\
\begin{pmatrix}
O\bigl(\log|z|\bigr) & O\bigl(|z|^\alpha\bigr)
\\
O\bigl(\log|z|\bigr) & O\bigl(|z|^\alpha\bigr)
\end{pmatrix}\begin{pmatrix}
1 & 0 \\ -d_\alpha-\alpha & 1
\end{pmatrix}&\text{if } -1<\alpha<0,
\end{cases}
\ee
where $d_\alpha$ is defined in~\eqref{eq:C}.
\item[(d)] There exists a $2\times 2$ matrix $Y_1$, independent of $z$, such that
\be
\label{eq:Yasympinfty}
Y(z)\sim I+Y_1z^{-1}+O(z^{-2})
\ee
as $z\to\infty$, uniformly within $\mathbb C\setminus[0,+\infty)$.
\end{itemize}
\end{RH}

Let us also introduce the resolvent operator $\mathcal R_{ x , t }^\sigma$, for all $x,t$ such that $Q_\sigma(x,t)\not=0$, by
\be\label{eq:resolventoperator}
\mathcal R_{ x , t }^\sigma=(1-\mathcal K^\sigma_{ x , t })^{-1}-1.
\ee

\begin{proposition}\label{prop:new}
The following facts hold true for all $x,t$ such that $Q_\sigma(x,t)\not=0$.

\textit{(i)}
$Y(z)$ defined in~\eqref{eq:Yexplicit} is the unique solution to Riemann--Hilbert problem~\ref{RHP:Y}.

\textit{(ii)}
The resolvent operator~\eqref{eq:resolventoperator} is also 2-integrable, acting on ${L}^2(0,+\infty)$ through the kernel
\be
\label{eq:resolvent}
R_{x,t}^\sigma(\lambda,\mu) = \frac{\f(\lambda)^\top Y(\lambda)^\top Y(\mu)^{-\top}\g(\mu)}{\lambda-\mu}.
\ee

\textit{(iii)}
In terms of the matrix $Y_1$ appearing in~\eqref{eq:Yasympinfty}, we have
\be
\label{eq:Ydet}
x\,\partial_x\log Q_\sigma(x,t) = (Y_1)_{1,2}.
\ee
\end{proposition}

\begin{remark}\label{rem:bdryvalues}
Equation~\eqref{eq:resolvent} involves the limiting values of $Y(z)$ as $z$ approaches the real axis.
Even though these limiting values depend on whether we approach the real axis from above or below, it follows from~\textit{(b)} in this Riemann--Hilbert problem that $Y_+\f = Y_-\f$ and ${Y_+}^{-\top}\mathbf g = {Y_-}^{-\top}\mathbf g$, such that this choice does not affect~\eqref{eq:resolvent}.
\end{remark}

\begin{proof}[Proof of Proposition~\ref{prop:new}]~

\textit{(i)} Let us first prove uniqueness.
To start, note that any solution $Y(z)$ has unit determinant: this follows by a standard argument based on the fact that $\det Y(z)$ extends analytically across $(0,+\infty)$ with possible isolated singularities at~$z=0,\widehat r_1,\dots,\widehat r_k$ by~\textit{(b)}, such singularities are removable by~\textit{(c)}, and $\det Y(z)$ tends to~$1$ as~$z\to \infty$ by~\textit{(d)}, such that Liouville theorem implies $\det Y(z)=1$ identically in~$z$.
Therefore, for any pair of solutions $Y(z),\widetilde Y(z)$ we can define $R(z)=\widetilde Y(z)Y(z)^{-1}$.
Then, $R(z)$ extends analytically across $(0,+\infty)$ with possible isolated singularities at~$z=0,\widehat r_1,\dots,\widehat r_k$ by~\textit{(b)} and such singularities are removable by~\textit{(c)}, as it is clear if $\alpha\geq 0$ while if $-1<\alpha<0$ it follows from
\begin{align}
\nonumber
R(z)&=\widetilde Y(z)\begin{pmatrix}1 & 0 \\ d_\alpha+\alpha & 1
\end{pmatrix}\begin{pmatrix}1 & 0 \\ -d_\alpha-\alpha & 1
\end{pmatrix}Y(z)^{-1}
\\
&= \begin{pmatrix}
O(\log |z|) & O(|z|^\alpha) \\ O(\log |z|) & O(|z|^\alpha)
\end{pmatrix}\begin{pmatrix}
O(|z|^\alpha) & O(|z|^\alpha) \\ O(\log |z|) &O(\log |z|) 
\end{pmatrix},
\end{align}
as $z\to 0$ (where we also use that $\det Y(z)=1$), such that all entries of $R(z)$ are $O(|z|^\alpha\,\log|z|)=o(z^{-1})$ as $z\to 0$.
Finally, $R(z)$ tends to~$I$ as $z\to \infty$ by~\textit{(d)}, such that Liouville theorem implies $R(z)=I$ identically in~$z$, i.e., $\widetilde Y(z)=Y(z)$ identically in~$z$.

Next, let us prove that $Y(z)$ as defined in~\eqref{eq:Yexplicit} solves Riemann--Hilbert problem~\ref{RHP:Y}.
Condition~\textit{(a)} is evident.
Next, note that~\eqref{eq:integrable} implies the general relation
\be
\mathcal K_{x,t}^\sigma\psi=-2\pi\mathrm i\, \mathcal C_-(\psi \mathbf g^\top)\mathbf f
\ee
such that we can compute
\be\label{eq:Yminusf}
Y_-\mathbf f = \mathbf f-2\pi \mathrm i\, \mathcal C_-(\mathbf F\, \mathbf g^\top)\mathbf f=\mathbf f+\mathcal K_{x,t}^\sigma \mathbf F=\mathbf F
\ee
by the definition~\eqref{eq:defF}.
By the Sokhotski--Plemelj formula,~$\mathcal C_+-\mathcal C_-$ is the identity operator, hence
\be
Y_+-Y_-=-2\pi\mathrm i\, \bigl(\mathcal C_+-\mathcal C_-\bigr)(\mathbf F\, \mathbf g^\top)=-2\pi\mathrm i\,\mathbf F\, \mathbf g^\top = -2\pi\mathrm i\,Y_-\mathbf f\, \mathbf g^\top 
\ee
which proves condition~\textit{(b)}.
The part of condition~\textit{(c)} relative to $z\to \widehat r_i$ follows from Lemma~\ref{lemma:F} and the properties of Cauchy operators recalled above, cf.~\eqref{cauchylog}.
For the part relative to $z\to 0$, note that $Y_{i,j}=\delta_{i,j}-\mathcal C(F_ig_j)$ and~\eqref{cauchyendpoint} is sufficient to prove~\eqref{eq:Y0} when $\alpha\geq 0$ (using the fact that $g_i,F_i$ are $O(\lambda^{\frac \alpha 2})$ as $\lambda\to 0_+$, cf.~Lemma~\ref{lemma:F}).
When $-1<\alpha<0$, note that $g_1(\lambda)+(d_\alpha+\alpha) g_2(\lambda)=O(\lambda^{\frac \alpha 2+1})$ as $\lambda\to 0_+$, hence
\begin{align}
\nonumber
Y(z)\begin{pmatrix}1 & 0 \\ d_\alpha+\alpha & 1
\end{pmatrix}
&=\begin{pmatrix}
1-\mathcal C\bigl(F_1(g_1+(d_\alpha+\alpha) g_2)\bigr)(z) & -\mathcal C\bigl(F_1g_2\bigr)(z) \\ 
-\mathcal C\bigl(F_2(g_1+(d_\alpha+\alpha) g_2)\bigr)(z) & 1-\mathcal C\bigl(F_2g_2\bigr)(z)
\end{pmatrix}
\\&=\begin{pmatrix}
O(\log|z|) & O(|z|^\alpha) \\ O(\log|z|) & O(|z|^\alpha)
\end{pmatrix}
\end{align}
which proves~\eqref{eq:Y0}, thanks to~\eqref{cauchyendpoint}, also in this case.
Finally, condition~\textit{(d)} follows from~\eqref{eq:Yexplicit} thanks to the fast decay of $\sigma$ at $+\infty$, with $Y_1$ given by
\be
\label{eq:Y1explicit}
Y_1 =\int_0^{+\infty}\mathbf F(\lambda)\,\mathbf g(\lambda)^\top\mathrm d\lambda.
\ee

\textit{(ii)} It follows from the general theory of integrable operators (cf.~\cite[Lemma~2.8]{DIZ}) that the resolvent $\mathcal R_{x,t}^\sigma$ is also 2-integrable, with kernel
\be
\label{eq:resolv}
R_{x,t}^\sigma (\lambda,\mu)=\frac{\mathbf F(\lambda)^\top\,\mathbf G(\mu)}{\lambda-\mu},
\ee
where $\mathbf F$ is in~\eqref{eq:defF} and $\mathbf G=(1-\mathcal K_{x,t}^\sigma)^{-1}\mathbf g$.
On the other hand, the analytic properties of $Y(z)^{-\top}$ imply the representation
\be
Y(z)^{-\top}=I+2\pi\mathrm i\,\mathcal C\bigl(Y_-^{-\top}\mathbf g\,\mathbf f^\top\bigr)(z).
\ee
(Namely, the right-hand side has the same jump across $(0,+\infty)$ and the same asymptotics at $\infty$ as $Y^{-\top}$, easily implying the claimed equality by standard complex analytic arguments.)
Therefore, noting that
\be
\mathcal K_{x,t}^\sigma\psi=2\pi\mathrm i\, \mathcal C_-(\psi \mathbf f^\top)\mathbf g,
\ee
we observe that
\be
Y_-^{-\top}\mathbf g = \mathbf g+2\pi \mathrm i\, \mathcal C_-(Y_-^{-\top}\mathbf g\, \mathbf f^\top)\mathbf g=\mathbf g+\mathcal K_{x,t}^\sigma(Y_-^{-\top}\mathbf g)
\ee
implying $Y_-^{-\top} \mathbf g=\mathbf G$.
Applying this last identity, as well as~\eqref{eq:Yminusf}, to~\eqref{eq:resolv} we conclude the proof.

\textit{(iii)} We start by recalling~\eqref{eq:Jacobifinal} and noting that the operator $\partial_x\widetilde {\mathcal K}_{x,t}^\sigma$ has kernel
\be
\frac x2 \sqrt{\sigma(\lambda)}\mathrm J_\alpha(x\sqrt{\lambda-t})\mathrm J_\alpha(x\sqrt{\mu-t}) \sqrt{\sigma(\mu)}= -x \widetilde f_1(\lambda)\widetilde g_2(\mu),
\ee
where $\widetilde f_1(\lambda)=f_1\bigl(x^2(\lambda-t)\bigr)$ and $\widetilde g_2(\mu)=g_2\bigl(x^2(\mu-t)\bigr)$.
Therefore, applying~\eqref{eq:Jacobifinal} and computing explicitly the trace we get
\begin{align}
\nonumber
\partial_x\log Q_\sigma(x,t) &= x\int_t^{+\infty}\bigl((1-\widetilde {\mathcal K}_{x,t}^{\sigma})^{-1}\widetilde f_1\bigr)(\lambda)\,\widetilde g_2(\lambda)\mathrm d\lambda
\\
&=\frac 1x\int_0^{+\infty}F_1(\lambda)g_2(\lambda)\mathrm d\lambda=(Y_1)_{1,2}
\end{align}
where in the last step we use~\eqref{eq:Y1explicit}.
\end{proof}

\subsection{Bessel dressing}
We now perform an explicit transformation of Riemann--Hilbert problem~\ref{RHP:Y}, by a standard procedure of~\emph{dressing}.
This will be useful in order to derive differential equations which $Q_\sigma( x , t )$ satisfies.

Introduce the matrix function
\be
\label{eq:PhiBe}
\Phi(z)=
\frac {\sqrt\pi}2 \,C
\begin{pmatrix}
2\,c_\alpha^{-1}\,\mathrm{J}_\alpha(\sqrt z) & -\mathrm{i}\,c_\alpha\,\mathrm{H}_\alpha^{(1)}(\sqrt z) \\
2\,c_\alpha^{-1}\,\sqrt z \,\mathrm{J}_\alpha'(\sqrt z) & -\mathrm{i}\,c_\alpha\,\sqrt z\,\bigl(\mathrm{H}_\alpha^{(1)}\bigr)'(\sqrt z)
\end{pmatrix} ,\qquad z\in\mathbb C\setminus[0,+\infty) ,
\ee
with $C$ defined in~\eqref{eq:C} and
\be
\label{eq:calpha}
c_\alpha=\exp\biggl(\frac{2\alpha+1}4\pi\mathrm{i}\biggr) .
\ee
Here, the branches of the multivalued functions involved in this definition are taken as follows: $\sqrt z$ is the branch of the square root defined for~$z\in\mathbb C\setminus[0,+\infty)$ by $\sqrt z = \sqrt{|z|}\e^{\mathrm{i}\arg z/2}$ with $\arg z\in (0,2\pi)$; 
$\mathrm{H}_\alpha^{(1)}$ is the principal branch of the Hankel function of the first kind of order~$\alpha$ and $(\mathrm{H}_\alpha^{(1)})'$ is its derivative, analytic in $\mathbb C\setminus(-\infty,0]$;
similarly, we take the principal branch of the Bessel function $\mathrm{J}_\alpha$ and of its derivative $\mathrm{J}_\alpha'$, analytic in $\mathbb C\setminus(-\infty,0]$.
(Note that the argument of Bessel and Hankel functions appearing in this definition belong to the upper half-plane by our choice of branch for the square root.)

Standard properties of Bessel and Hankel functions (cf.~\cite[\S~10]{DLMF}) imply that $\Phi(z)$ is analytic in $z\in\mathbb C\setminus[0,+\infty)$, and the boundary values $\Phi_\pm(\lambda)=\lim_{\epsilon\to 0_+}\Phi(\lambda\pm\mathrm{i}\epsilon)$ are continuous for $\lambda\in(0,+\infty)$ and satisfy
\be
\label{eq:jumpPhi}
\Phi_+(\lambda) = \Phi_-(\lambda)\begin{pmatrix}
\e^{-\mathrm{i}\pi\alpha} & 1 \\ 0 & \e^{\mathrm{i}\pi\alpha} 
\end{pmatrix},
\quad \lambda \in (0,+\infty).
\ee
This follows from analytic continuation properties of Bessel and Hankel functions~\cite[\S~10.11]{DLMF}.
Moreover, for any $0<\delta<\pi$ we have
\be
\label{eq:PhiBeasympinfty}
\Phi(z) \sim \biggl(I + \Phi_1 z^{-1}+O(z^{-2})\biggl)z^{-\frac 14\sigma_3}\,G\,\e^{-\mathrm{i}\sqrt z\sigma_3}\times\begin{cases}
\left(\begin{smallmatrix}1 & 0 \\ \e^{-\mathrm{i}\pi\alpha} & 1
\end{smallmatrix}\right)&\mbox{if }0<\arg z<\delta, \\
I, & \delta<\arg z<2\pi-\delta, \\
\left(\begin{smallmatrix}1 & 0 \\ -\e^{\mathrm{i}\pi\alpha} & 1
\end{smallmatrix}\right) &\mbox{if }2\pi-\delta<\arg z<2\pi,
\end{cases}
\ee
as $z\to\infty$, uniformly within $\mathbb C\setminus[0,+\infty)$, where, $\sigma_3$ is the third Pauli matrix, $\sigma_3={\rm diag}(1,-1)$, the branch of $\sqrt z$ is as above and the branches of $z^{\pm\frac 14}=|z|^{\pm\frac 14}\e^{\pm\frac {\mathrm{i}} 4\arg z}$ are defined similarly by $\arg z\in (0,2\pi)$, and where
\be
\label{eq:Phi1BeG}
\Phi_1 =\begin{pmatrix}
 \frac{-16 \alpha^4+40 \alpha^2-9}{128}  & \frac{4 \alpha^2-1}{8} \\
 \frac{-64 \alpha^6+368 \alpha^4-556 \alpha^2+117}{1536} & \frac{16 \alpha^4-40 \alpha^2+9}{128} \\
\end{pmatrix},
\qquad G=\frac 1{\sqrt 2}\begin{pmatrix} 1 & -\mathrm i \\ -\mathrm i & 1\end{pmatrix} .
\ee
This follows from asymptotic properties of Bessel and Hankel functions~\cite[\S~10.17]{DLMF}.
Finally, we have the symmetry property
\be
\label{eq:symmetryPhi}
\overline{\Phi(\overline z)}
=\Phi(z)\,\e^{\mathrm{i}\frac \pi 2\sigma_3},\quad z\in \mathbb C\setminus[0,+\infty).
\ee
This follows from analytic continuation properties of Bessel and Hankel functions~\cite[\S~10.11]{DLMF}, as well as from our choice of the square root branch, satisfying $\sqrt {\overline z}=-\overline{\sqrt z}$ for all $z\in\mathbb C\setminus[0,+\infty)$.

\begin{remark}
\label{remark:reasonforC}
The reason for introducing the matrix factor $C$ in~\eqref{eq:PhiBe} and~\eqref{eq:fghat} is to have the asymptotics at infinity normalized as in~\eqref{eq:PhiBeasympinfty}.
Hovewever, this choice is quite inessential.
\end{remark}

The following identities will be useful later.

\begin{proposition}
\label{prop:identitiesfg}
We have, for all $\lambda\in(0,+\infty)$ and with $c_\alpha$ defined in~\eqref{eq:calpha},
\be
\label{eq:identitiesfg}
\widehat \f(\lambda)=\frac{c_\alpha}{\sqrt \pi} \Phi_+(\lambda) \begin{pmatrix}
1 \\ 0
\end{pmatrix},
\quad
\widehat \g(\lambda)= -\frac{c_\alpha}{\sqrt \pi} \Phi_+(\lambda)^{-\top} \begin{pmatrix}
0\\ 1\end{pmatrix}.
\ee
\end{proposition}
\begin{proof}
Use~\eqref{eq:PhiBe} and the fact that $\det\Phi(z)=1$ identically in~$z$,~cf.~\cite[\S~10.5]{DLMF}
\end{proof}

\begin{corollary}
\label{corollary:jump}
Let $J_Y(\lambda)$ be as in~\eqref{eq:Yjump}. Then
\be
J_Y(\lambda)=\Phi_-(\lambda)\begin{pmatrix}
\e^{-\mathrm{i}\pi\alpha} & 1-\sigma(x^{-2}\lambda+t) \\ 0 & \e^{\mathrm{i}\pi\alpha} 
\end{pmatrix}\Phi_+(\lambda)^{-1}.
\ee
\end{corollary}
\begin{proof}
It follows from Proposition~\ref{prop:identitiesfg} combined with~\eqref{eq:jumpPhi}. (Note that $\mathrm{i}\e^{-\mathrm{i}\pi\alpha}c_\alpha^2=-1$.)
\end{proof}

As anticipated, we now introduce a transformation (\emph{dressing}) of the matrix $Y$:
\be
\label{eq:Psi}
\Psi(z) =  \begin{pmatrix}
1 & 0 \\ -\frac {xt}2 & 1
\end{pmatrix}\,x^{\frac 12\sigma_3}\,Y\bigl(x^2(z-t)\bigr)\,\Phi\bigl(x^2(z-t)\bigr).
\ee
The Riemann--Hilbert conditions determining uniquely $Y$ change as follows.

\begin{RH}\label{RHP:Psi}~
\begin{itemize}
\item[(a)] $\Psi(z)$ is a $2\times 2$ matrix analytically depending on $z\in\mathbb C\setminus[t ,+\infty)$, for all $ x >0, t \in\mathbb R$.
\item[(b)] For all $\lambda\in(t,+\infty)\setminus\lbrace r_1,\dots, r_k\rbrace$, the boundary values $\Psi_\pm(\lambda)=\lim_{\epsilon\to 0_+}\Psi(\lambda\pm\mathrm{i}\epsilon)$ exist, are continuous, and satisfy
\be
\label{eq:Psijump}
\Psi_+(\lambda) = \Psi_-(\lambda) J_\Psi(\lambda),\quad J_\Psi(\lambda)=\begin{pmatrix}
\e^{-\mathrm{i}\pi\alpha} & 1-\sigma(\lambda) \\ 0 & \e^{\mathrm{i}\pi\alpha} 
\end{pmatrix}.
\ee
\item[(c)] As $z\to r_j$ away from the real axis we have $\Psi(z)=O\bigl(\log|z-r_j|\bigr)$, and as  $z\to t$ away from $(t,+\infty)$ we have
\be
\Psi(z)=\begin{cases}
O\bigl(\log|z-t|\bigr)\Phi\bigl(x^2(z-t)\bigr)&\mbox{if }\alpha\geq 0,\\
\begin{pmatrix}
O\bigl(\log|z-t|\bigr) & O\bigl(|z-t|^\alpha\bigr)
\\
O\bigl(\log|z-t|\bigr) & O\bigl(|z-t|^\alpha\bigr)
\end{pmatrix}\begin{pmatrix}
1 & 0 \\ -d_\alpha-\alpha & 1
\end{pmatrix}\Phi\bigl(x^2(z-t)\bigr)&\mbox{if } -1<\alpha<0,
\end{cases}
\ee
where $d_\alpha$ is defined in~\eqref{eq:C}.
\item[(d)] There exists a $2\times 2$ matrix $\Psi_1$, depending on $ x , t $ only, such that
\be
\label{eq:Psiasympinfty}
\Psi(z) \sim \biggl(I+\Psi_1 z^{-1}+O(z^{-2})\biggl)z^{-\frac{\sigma_3}4}G\e^{-\mathrm{i} x \sqrt z\sigma_3}\times\begin{cases}
\left(\begin{smallmatrix}1 & 0 \\ \e^{-\mathrm{i}\pi\alpha} & 1
\end{smallmatrix}\right) & 0<\arg z<\delta, \\
I, & \delta<\arg z<2\pi-\delta, \\
\left(\begin{smallmatrix}1 & 0 \\ -\e^{\mathrm{i}\pi\alpha} & 1
\end{smallmatrix}\right) & 2\pi-\delta<\arg z<2\pi,
\end{cases}
\ee
as $z\to\infty$, uniformly within $\mathbb C\setminus[0,+\infty)$, with the notations of~\eqref{eq:PhiBeasympinfty}.
\end{itemize}
\end{RH}

It is straightforward to show that the matrix $\Psi_1$ appearing in~\eqref{eq:Psiasympinfty} is given by
\be
\Psi_1=\frac 1{x^2}B(Y_1+\Phi_1)B^{-1}+
\begin{pmatrix}
 -\frac{t (x^2t -2)}{8} & -\frac{xt}{2} \\
 \frac{xt^2(x^2t -3)}{24} & \frac{t (x^2t -2)}{8} 
\end{pmatrix},\qquad B=\begin{pmatrix}
1 & 0 \\ -\frac {xt}2 & 1
\end{pmatrix}x^{\frac 12\sigma_3}.
\ee
In particular, by~\eqref{eq:Phi1BeG} we have
\be
\label{eq:vpaxQ}
\bigl(\Psi_1\bigr)_{1,2}=\frac 1x\bigl(Y_1\bigr)_{1,2}-\frac {xt}2+\frac{4\alpha^2-1}{8x}=\partial_x\log Q_\sigma(x,t)-\frac {xt}2+\frac{4\alpha^2-1}{8x}
\ee

\begin{remark}
Note that the jump matrix $J_\Psi$ does not depend on the parameters $ x >0,\,t \in\mathbb R$.
This is the main advantage of the transformation from $Y$ to $\Psi$.
\end{remark}

\begin{proposition}
The matrix function $\Psi$ defined in~\eqref{eq:Psi} is the unique solution to Riemann--Hilbert problem~\ref{RHP:Psi} and it satisfies
\be
\label{eq:symmetry}
\overline{\Psi(\overline z)}=\Psi(z)\,\e^{\mathrm{i}\frac{\pi}2\sigma_3},\quad z\in\mathbb C\setminus[t,+\infty).
\ee
\end{proposition}

\begin{proof}
All the Riemann--Hilbert conditions for~$\Psi$ follow easily from the definition~\eqref{eq:Psi}.
In particular, condition~\textit{(b)} follows from Corollary~\ref{corollary:jump}.
Moreover, uniqueness of the solution follows from the uniqueness of the solution to the Riemann--Hilbert problem for~$Y$: indeed, different solutions~$\Psi$ could be used via~\eqref{eq:Psi} to define different solutions~$Y$.

Finally, it can be checked by using the conditions~\textit{(a)}--\textit{(d)} in Riemann--Hilbert problem~\ref{RHP:Psi} and the identity~\eqref{eq:symmetryPhi}, that $\overline{\Psi(\overline z)}\,\e^{-\mathrm{i}\frac{\pi}2\sigma_3}$ satisfies the exact same conditions, whence~\eqref{eq:symmetry} follows from the uniqueness of the solution.
\end{proof}

In the next proposition, we look more carefully into the behavior of~$\Psi(z)$ when $z\to t $ and when $z\to r_j$.

\begin{proposition}
\label{proposition:behavioratzero}
There exists $\rho>0$ sufficiently small such that the following statements hold.

\textit{(i)} We have
\be
\label{eq:behavioratzero}
\Psi(z) = E(z)\begin{pmatrix}
1 & h(z; t ) \\ 0 & 1
\end{pmatrix}(z- t)^{\frac\alpha 2\sigma_3}
\ee
where $E(z)$ is analytic in $z$ for $|z-t|<\rho$, and
\be 
\label{eq:h}
h(z;t)=\frac 1{2\pi\mathrm{i}}\int_0^{\rho} u^{\alpha}\bigl(\e^{\mathrm{i}\pi\alpha}+\mathrm{i}\sigma(u+t)\bigr)\frac{\mathrm{d} u}{u+t-z} .
\ee
In~\eqref{eq:behavioratzero}, we take the branch of $(z-t)^{\frac\alpha  2\sigma_3}$ which is discontinuous on $[t,+\infty)$ and with positive limiting values when $z$ approaches $( t ,+\infty)$ from above.

\textit{(ii)}
For all $j=1,\dots,k$ such that $r_j>t$, we have
\be
\label{eq:behavioratrj}
\Psi(z) = E_j(z)\begin{pmatrix}
1 & h_j(z) \\ 0 & 1
\end{pmatrix}\e^{\mp \mathrm i\pi\frac{\alpha}2\sigma_3},\quad \pm\Im z>0,
\ee
where $E_j(z)$ is analytic for $|z-r_j|<\rho$ and $h_j(z)=\frac 1{2\pi\mathrm i}\int_{r_j-\rho}^{r_j+\rho}\frac{1-\sigma(\lambda)}{\lambda-z}\mathrm d\lambda$.
\end{proposition}
\begin{proof}~

\textit{(i)}
Let $\widehat E(z)=\Psi(z)(z- t )^{-\frac\alpha  2\sigma_3}$.
It follows from condition~\textit{(b)} in Riemann--Hilbert problem~\ref{RHP:Psi} that the boundary values $\widehat E_\pm(\lambda)=\lim_{\epsilon\to 0_+}\widehat E(\lambda\pm\mathrm{i}\epsilon)$ at $\lambda\in( t ,+\infty)$ are continuous and related as
\be
\widehat E_+(\lambda)=
\widehat E_-(\lambda)
\begin{pmatrix}
1 & (\lambda-t)^\alpha\bigl(\e^{\mathrm{i}\pi\alpha}+\mathrm{i}\sigma(\lambda)\bigr) \\ 0 & 1
\end{pmatrix},\qquad \lambda\in(t,+\infty).
\ee
On the other hand, $h(z; t )$ in the statement is analytic for $z\in\mathbb C\setminus [ t , t +\rho]$ and the boundary values $h_\pm(\lambda; t )=\lim_{\epsilon\to 0_+}h(\lambda\pm\mathrm{i}\epsilon; t )$ are continuous in $\lambda\in( t , t +\rho)$ (provided $\rho$ is sufficiently small to avoid the discontinuities of $\sigma$) and related as
\be
h_+(\lambda; t )=h_-(\lambda; t )+(\lambda- t )^\alpha\bigl(\e^{\mathrm{i}\pi\alpha}+\mathrm{i}\sigma(\lambda)\bigr),\qquad \lambda\in( t , t +\rho),
\ee 
as it follows from the Sokhotski--Plemelj formula.
Combining these facts, and using Morera thoerem, we obtain that
\be
E(z)=\widehat E(z)\begin{pmatrix}
1 & -h(z; t ) \\ 0 & 1
\end{pmatrix}
\ee
extends analytically across $( t , t +\rho)$, and is therefore analytic in $z$ for $|z-t|<\rho$ (the isolated singularity at $z=t$ is removable). 

\textit{(ii)}
The proof is similar to the previous point.
We start by defining $\widehat E_j(z)=\Psi(z)\e^{\pm \mathrm i\pi\frac{\alpha}2\sigma_3}$ ($\pm\Im z>0$) for $|z-r_j|<\rho$ with $\rho>0$ sufficiently small to avoid other singularities of $\sigma$ as well as $t$.
Then $\widehat E_j$ has boundary values $\widehat E_{j,\pm}(\lambda)=\lim_{\epsilon\to 0_+}\widehat E_j(\lambda\pm\mathrm i\epsilon)$ which are continuous for $\lambda\in (r_j-\rho,r_j+\rho)\setminus\lbrace r_j\rbrace$ and satisfy
\be
\widehat E_{j,+}(\lambda)=
\widehat E_{j,-}(\lambda)\begin{pmatrix}
1 & 1-\sigma(\lambda) \\ 0 & 1
\end{pmatrix},\qquad  \lambda\in (r_j-\rho,r_j+\rho)\setminus\lbrace r_j\rbrace.
\ee
On the other hand, $h_j(z)$ in the statement is analytic for $z\in\mathbb C\setminus [ r_j-\rho,r_j +\rho]$ and the boundary values $h_{j,\pm}(\lambda)=\lim_{\epsilon\to 0_+}h_j(\lambda\pm\mathrm{i}\epsilon)$ are continuous in $\lambda\in(r_j-\rho,r_j+\rho)\setminus\lbrace r_j\rbrace$ and related as
\be
h_{j,+}(\lambda)=h_{j,-}(\lambda)+1-\sigma(\lambda),\qquad \lambda\in(r_j-\rho,r_j+\rho)\setminus\lbrace r_j\rbrace,
\ee 
as it follows from the Sokhotski--Plemelj formula.
Combining these facts, and using Morera thoerem, we obtain that
\be
E_j(z)=\widehat E_j(z)\begin{pmatrix}
1 & -h_j(z) \\ 0 & 1
\end{pmatrix}
\ee
extends analytically across $(r_j-\rho,r_j+\rho)$, and is therefore analytic in $z$ for $|z-r_j|<\rho$ (the isolated singularity at $z=r_j$ is removable). 
\end{proof}

For later convenience, we report following estimates as $z\to t$, where $h(z;t)$ is given in~\eqref{eq:h}:
\be
\label{eq:asymph}
h(z;t)=o\biggl(\frac 1{z-t}\biggr),\qquad \partial_zh(z;t)=O\biggl(\frac 1{z-t}\biggr),\qquad \partial_th(z;t)=O\biggl(\frac 1{z-t}\biggr).
\ee
These asymptotics follow from general properties of Cauchy integrals near the endpoints of integration contour, cf.~\cite{Musk}.

\section{Differential equations}
\subsection{The Lax pair}
Thanks to~\eqref{eq:vpaxQ} we can express $v=v(x,t)$ defined in~\eqref{eq:defvstatement} as 
\be
\label{eq:v}
v( x , t )=-\bigl(\Psi_1\bigr)_{1,2}.
\ee
Let us introduce
\be
\label{eq:Theta}
\Theta(z)=\begin{pmatrix}
1 & 0 \\ v( x , t ) & 1
\end{pmatrix}\Psi(z) .
\ee
\begin{proposition}
\label{prop:equations}
Denoting $v=v(x,t)$, we have
\begin{align}
\label{eq:diffeqx}
\partial_x\Theta(z)&=\begin{pmatrix} 0 & 1\\ -z+2\,v_x & 0
\end{pmatrix}\Theta(z),\\
\label{eq:diffeqt}
\partial_t\Theta(z)&=\biggl[\begin{pmatrix}0&0\\ v_t & 0
\end{pmatrix}+\frac{1}{z-t}V(x,t)
\biggr]\Theta(z),
\end{align}
where
\be
\label{eq:S}
V(x,t)=\begin{pmatrix}
\frac 12v_{xt} & -v_t \\  t v_ t-2v_ t v_ x+\frac 12 v_{xxt} &-\frac 12v_{xt}
\end{pmatrix}.
\ee
Moreover, $v$ satisfies
\be
\label{eq:differ}
\bigl(t-2v_x\bigr)v_{xt}+\frac 14v_{xxxt}-v_{xx}v_t=0,
\ee
and
\be
\label{eq:PDE}
(2v_x-t){v_t}^2+\frac 14{v_{xt}}^2-\frac 12v_tv_{xxt}=\frac{\alpha^2}4.
\ee
\end{proposition}
\begin{remark}
Equation~\eqref{eq:differ} is (up multiplication by~$v_t$) the $x$-derivative of equation~\eqref{eq:PDE}.
\end{remark}
\begin{proof}
We follow a standard argument based on the fact that $J_\Psi$ does not depend on the parameters~$ x , t $, cf.~\cite[Section~3]{CCR} for more details.

The ratio $L(z)=\partial_ x \Theta(z)\cdot\Theta^{-1}(z)$, as a function of~$z$, has no discontinuity across $( t ,+\infty)$ and is~$O(1)$ at~$z= t $.
It follows that $L(z)$ is entire in $z$.
As such, it can be determined by the asymptotics at~$z=\infty$.
The latter are explicit in terms of~\eqref{eq:Psiasympinfty}, and, after a simple calculation, we get
\be
\label{eq:rdef}
L(z)=\begin{pmatrix}
0&1\\-z+2r+v^2+ v_x & 0
\end{pmatrix},
\qquad
r=\bigl(\Psi_1\bigr)_{1,1}.
\ee
Moreover, the $(1,2)$-entry of the coefficient of $z^{-1}$ in the Laurent series of $L(z)$ at $z=\infty$ can also be computed using~\eqref{eq:Psiasympinfty} to be $2r+v^2-v_x$; however, we just proved that $L(z)$ is a polynomial in $z$, whence $2r+v^2=v_x$ and the proof of~\eqref{eq:diffeqx} is complete.

Similarly, the ratio $M(z)=\partial_ t \Theta(z)\cdot\Theta^{-1}(z)$, as a function of~$z$, has no discontinuity across $( t ,+\infty)$ and is $O\bigl((z- t )^{-1}\bigr)$ at~$ t $, as it follows from Proposition~\ref{proposition:behavioratzero} and the estimates~\eqref{eq:asymph}.
Hence, $M(z)$ is meromorphic in~$z$ with a unique pole at~$z= t $, which is simple:
by Liouville theorem and~\eqref{eq:Psiasympinfty} we get
\be
M(z)=\begin{pmatrix}0&0\\ v_t & 0\end{pmatrix}
+\frac 1{z-t}\begin{pmatrix}a & b \\ c & -a\end{pmatrix},
\ee
for certain $a=a( x , t )$, $b=b( x , t )$, and $c=c( x , t )$.
We use here that $M(z)$ is traceless because $\det\Theta(z)=1$ identically in~$z$.
Equating cross derivatives $\partial_ x \partial_ t \Theta=\partial_ t \partial_ x \Theta$ we obtain the \textit{zero-curvature equation}
\be
\partial_ t  L-\partial_ x  M=ML-LM.
\ee
Expanding this equation at $z=\infty$ we obtain the following four identities:
\be
b+v_t=0,\quad
2\,a+b_x=0,\quad
2\,b\,v_x+a_x- t \, b-c=0,\quad
4\,a\,v_x-2\, t \,a-c_ x=0.
\ee
We use the first three to determine $a,b,c$ in terms of~$v$ and its derivatives, as claimed in~\eqref{eq:S}. 
The fourth relation implies~\eqref{eq:differ}.

Finally, we note that $\tr(M(z)^2)=\tr\bigl(\partial_t^2\Theta(z)\cdot\Theta(z)^{-1}\bigr)$.
We can expand this identity near $z= t $: using the explicit formula for $M$ in~\eqref{eq:diffeqt} and~\eqref{eq:S} we can expand the left-hand side as
\be
\frac 1{(z- t )^2}\biggl[\bigl(4v_x-2 t \bigr){v_t}^2+\frac 12{v_{xt}}^2-v_{xxt}v_t\biggr]+O\bigl((z- t )^{-1}\bigr),
\ee
while expanding the right-hand side directly using Proposition~\ref{proposition:behavioratzero} and the estimates~\eqref{eq:asymph} we get
\be
\frac {\alpha^2}{2\,(z- t )^2}+O\bigl((z- t )^{-1}\bigr).
\ee
Comparing the last two expansions proves~\eqref{eq:PDE}.
\end{proof}

Introduce
\be
\label{eq:varphi1}
\varphi(z;x,t)=\frac 1{\sqrt{2\pi}}\bigl(\Theta(z)\bigr)_{1,1}=\frac 1{\sqrt{2\pi}}\bigl(\Psi(z)\bigr)_{1,1}.
\ee
It follows from~\eqref{eq:diffeqx} that
\be
\label{eq:varphi2}
\bigl(\Theta(z)\bigr)_{2,1}=\sqrt{2\pi}\,\partial_ x \varphi(z;x,t)
\ee
and that $\varphi$ solves
\be
\label{eq:Schrodinger}
\bigl(-\partial_x^2+2v_x\bigr)\varphi(z;x,t)=z\,\varphi(z;x,t).
\ee
Note also that $\varphi(z)$ is analytic in $z\in\mathbb C\setminus[ t ,+\infty)$, with monodromy around the origin described by the jump relation
\be
\label{eq:monodromyvarphi}
\varphi_+(\lambda;x,t)=\e^{-\mathrm{i}\pi\alpha}\varphi_-(\lambda;x,t),\qquad\varphi_\pm(\lambda;x,t)=\lim_{\epsilon\to 0_+}\varphi(\lambda\pm\mathrm{i}\epsilon;x,t),\qquad \lambda\in ( t ,+\infty),
\ee
as it follows from~\textit{(b)} in Riemann--Hilbert problem~\ref{RHP:Psi}.
Moreover, it follows from~\eqref{eq:symmetry} that
\be
\overline{\varphi(\overline z;x,t)}=\mathrm{i}\,\varphi(z).
\ee
Combined with~\eqref{eq:monodromyvarphi}, we get
\be
\label{eq:conjugate}
\overline{\varphi_+(\lambda;x,t)}=\lim_{\epsilon\to 0_+}\overline{\varphi(\lambda+\mathrm{i}\epsilon;x,t)}
=\lim_{\epsilon\to 0_+}\mathrm{i}\,\varphi(\lambda-\mathrm{i}\epsilon;x,t)=\mathrm{i}\,\varphi_-(\lambda;x,t)=\mathrm{i}\,\e^{\mathrm{i}\pi\alpha}\,\varphi_+(\lambda;x,t).
\ee

\subsection{Integral relations}

The jump matrix $J_\Psi(\lambda)$ is not piecewise constant in~$\lambda$ and so we cannot complement the Lax pair of the last section with an ordinary differential equation in $z$ for $\Theta(z)$ with meromorphic coefficients, cf.~\cite{CD}.
However, we can get additional information for our purposes as follows.
Let us start with a simple consequence of Proposition~\ref{proposition:behavioratzero}.

\begin{proposition}
\label{prop:pawpaeta}
As $z\to t $, we have
\be
\partial_z\Theta(z)\cdot\Theta(z)^{-1}+\partial_t\Theta(z)\cdot\Theta(z)^{-1}=o\biggl(\frac1{z-t}\biggr).
\ee
\end{proposition}
\begin{proof}
With the notation of Proposition~\ref{proposition:behavioratzero}, first point, we have $\Theta=\widetilde E\begin{pmatrix}
1 & h \\ 0 & 1
\end{pmatrix}(z- t )^{\frac\alpha  2\sigma_3}$ with $\widetilde E=\begin{pmatrix}
1 & 0 \\ v & 1 \end{pmatrix}E$, such that we can compute
\begin{align}
\Theta_z\,\Theta^{-1}&=\widetilde E_z\, \widetilde E^{-1}
+\widetilde E\begin{pmatrix}
0 & h_z \\ 0 & 0 
\end{pmatrix}\widetilde E^{-1}
+\widetilde E\begin{pmatrix}
1 & h \\ 0 & 1
\end{pmatrix}
\frac{\tfrac\alpha 2\sigma_3}{z- t }
\begin{pmatrix}
1 & -h \\ 0 & 1
\end{pmatrix}\widetilde E^{-1},
\\
\Theta_t\,\Theta^{-1}&=\widetilde E_t\,\widetilde E^{-1}
+\widetilde E\begin{pmatrix}
0 & h_t \\ 0 & 0 
\end{pmatrix}\widetilde E^{-1}-
\widetilde E
\begin{pmatrix}
1 & h \\ 0 & 1
\end{pmatrix}\frac{\tfrac\alpha  2\sigma_3}{z- t }\begin{pmatrix}
1 & -h \\ 0 & 1
\end{pmatrix}\widetilde E^{-1}.
\end{align}
Therefore the proof is complete by adding these expressions and using the following estimates as $z\to t $:
\begin{align}
\widetilde E_z\widetilde E^{-1}=O(1),\quad
\widetilde E_t\widetilde E^{-1}=O(1),\quad
h_z+h_t=\int_0^\rho\mathrm{i} u^\alpha\sigma'(u+t)\frac{\mathrm{d} u}{u+ t -z}=o\biggl(\frac1{z-t}\biggr),
\end{align}
where the last one follows from general properties of Cauchy integrals near the endpoints of the integration contour, cf.~\cite[Chapter~4]{Musk}.
\end{proof}

\begin{proposition}
\label{prop:nonlocal}
In terms of~\eqref{eq:varphi1} and $c_\alpha$ defined in~\eqref{eq:calpha}, we have
\begin{align}
\label{eq:dtlogF}
v_t-\frac x2&=c_\alpha^2\int_ t ^{+\infty}\varphi_+(\lambda;x,t)^2\,\mathrm d\sigma(\lambda),
\\
\label{eq:dxlogF}
x  t - (x v)_x&=2c_\alpha^2\int_ t ^{+\infty}(\lambda- t )\,\varphi_+(\lambda;x,t)^2\,\mathrm d\sigma (\lambda).
\end{align}
\end{proposition}

We recall that the Stieltjes integrals with respect to $\mathrm d\sigma$ have been defined in the Introduction, cf.~\eqref{eq:dsigma}.

\begin{proof}
Let us denote $A(z)=\partial_z\Theta(z)\cdot\Theta(z)^{-1}$, which is an analytic function of $z\in\mathbb C\setminus[ t ,+\infty)$ with boundary values $A_\pm(\lambda)=\lim_{\epsilon\to 0_+}A(\lambda\pm\mathrm{i}\epsilon)$ continuous in $(t,+\infty)\setminus\lbrace r_1,\dots,r_k\rbrace$ and related by
\begin{align}
\nonumber
A_+(\lambda)-A_-(\lambda)&=\Theta_+(\lambda)\cdot J_\Psi(\lambda)^{-1}\cdot\partial_\lambda J_\Psi(\lambda)\cdot\Theta_+(\lambda)^{-1}
\\
&=-\e^{\mathrm{i}\pi\alpha}\,\Theta_+(\lambda)\cdot\begin{pmatrix}0 & \sigma'(\lambda) \\ 0 & 0
\end{pmatrix}\cdot\Theta_+(\lambda)^{-1},
\end{align}
for all $\lambda\in(t,+\infty)\setminus\lbrace r_1,\dots,r_k\rbrace$, as it follows from~\textit{(b)} in Riemann--Hilbert problem~\ref{RHP:Psi}.
Therefore, a contour deformation argument yields
\begin{align}
\nonumber
&\frac{c_\alpha^2}{2\pi}\int_{(t,+\infty)\setminus\lbrace r_1,\dots,r_k\rbrace}\Theta_+(\lambda)\cdot\begin{pmatrix}0 & \sigma'(\lambda) \\ 0 & 0
\end{pmatrix}\cdot\Theta_+(\lambda)^{-1}\mathrm d\lambda
\\
\nonumber
&\qquad =
\frac 1{2\pi\mathrm{i}}\int_{(t,+\infty)\setminus\lbrace r_1,\dots,r_k\rbrace}\bigl(A_+(\lambda)-A_-(\lambda)\bigr)\mathrm d\lambda
\\
&\qquad=\res{z=t}A(z)\mathrm{d} z+\res{z=\infty}A(z)\mathrm{d} z+\sum_{j=1}^k1_{r_j>t}\res{z=r_j}A(z)\mathrm d z.
\end{align}
Here, the residues are meant as \textit{formal residues}: namely, the formal residue at $ t , r_j$ is $1/(2\pi\mathrm{i})$ times the limit of the contour integral over a positively oriented circle centered at~$ t,r_j $ as the radius tends to zero, and the formal residue at $\infty$ is $1/(2\pi\mathrm{i})$ times the limit of the contour integral over a negatively oriented circle as the radius diverges to $+\infty$.
As long as the asymptotic expansion around these singularities is uniform away from $[ t ,+\infty)$ (which will be the case), such formal residues can be computed as if they were ordinary residues by extracting the appropriate coefficient in the asymptotic expansions of the integrand around the singularity: for the residue at~$ t $ this is the coefficient of $(z- t )^{-1}$, for that at $r_j$ the coefficient of $(z-r_j)^{-1}$, and for that at~$\infty$ the coefficient of $z^{-1}$ taken with opposite sign.

Take the $(1,2)$-entry of the last identity to get, using~\eqref{eq:varphi1},
\begin{align}
\nonumber
&c_\alpha^2\int_{(t,+\infty)\setminus\lbrace r_1,\dots,r_k\rbrace}\varphi_+(\lambda;x,t)^2\,\sigma'(\lambda)\,\mathrm d\lambda
\\
&\quad=\res{z= t }\bigl(A(z)\bigr)_{1,2}\,\mathrm{d} z+\res{z=\infty}\bigl(A(z)\bigr)_{1,2}\,\mathrm{d} z+\sum_{j=1}^k1_{r_j>t}\res{z=r_j}\bigl(A(z)\bigr)_{1,2}\,\mathrm{d} z.
\end{align}
The residue at $\infty$ is computed to be~$-x/2$, thanks to~\eqref{eq:Psiasympinfty}, and the one at~$t$ to be~$v_t$, thanks to Proposition~\ref{prop:pawpaeta} and~\eqref{eq:diffeqt}, and those at $r_j$ are computed as follows.
Recall Proposition~\ref{proposition:behavioratzero} and set $\widetilde E_j=\begin{pmatrix} 1 & 0 \\ v & 1 \end{pmatrix}E_j$, such that
\begin{align}
\nonumber
\res{z=r_j}\bigl(A(z)\bigr)_{1,2}\,\mathrm d z&=\res{z=r_j}\biggl(\widetilde E_j(z)\begin{pmatrix}
0 & \partial_z h_j(z) \\ 0 & 0
\end{pmatrix}\widetilde E_j(z)^{-1}\biggr)_{1,2}\,\mathrm d z
\\
\label{verylast}
&=\res{z=r_j}\partial_z h_j(z)\mathrm d z\,\bigl(\widetilde E_j(r_j)_{1,1}\bigr)^2
\end{align}
where we recall that $\widetilde E_j(z)$ is analytic in $z$ for $z$ sufficiently close to $r_j$ and has unit determinant.
By the definition of $\widetilde E_j$, we have
\be
\label{veryverylast}
\bigl(\widetilde E_j(r_j)_{1,1}\bigr)^2=2\pi\varphi_+(r_j;x,t)^2\,\e^{\mathrm i\pi\alpha}.
\ee 
It remains to compute the formal residue $\res{z=r_j}\partial_z h_j(z)\mathrm d z$. To this end, introduce
\be
\widehat\sigma(\lambda)=\begin{cases}\sigma(\lambda), &\mbox{if }\lambda>r_j,
\\
\lim_{\epsilon\to 0_+}\sigma(r_j+\epsilon), &\mbox{if }\lambda=r_j,
\\
\sigma(\lambda)+m_j,&\mbox{if }\lambda<r_j,
\end{cases}
\ee
which is smooth in a neighborhood of $\lambda=r_j$, where we recall that $m_j$ is defined just after~\eqref{eq:dsigma}.
Therefore
\begin{align}
\nonumber
h_j(z)&=\frac 1{2\pi\mathrm i}\int_{r_j-\rho}^{r_j+\rho}\frac{1-\widehat{\sigma}(\lambda)}{\lambda-z}\mathrm d\lambda+\frac{m_j}{2\pi\mathrm i}\int_{r_j-\rho}^{r_j}\frac{\mathrm d\lambda}{\lambda-z}
\\
& =\frac 1{2\pi\mathrm i}\int_{r_j-\rho}^{r_j+\rho}\frac{1-\widehat{\sigma}(\lambda)}{\lambda-z}\mathrm d\lambda+\frac{m_j}{2\pi\mathrm i}\log(z-r_j)-\frac{m_j}{2\pi\mathrm i}\log(z-r_j+\rho)
\end{align}
where we use the principal branch of the logarithm.
Therefore,
\be
\partial_z h_j(z) = \frac 1{2\pi\mathrm i}\int_{r_j-\rho}^{r_j+\rho}\frac{1-\widehat{\sigma}(\lambda)}{(\lambda-z)^2}\mathrm d\lambda+\frac{m_j}{2\pi\mathrm i(z-r_j)}-\frac{m_j}{2\pi\mathrm i(z-r_j+\rho)}
\ee
and the first term in the right-hand side is bounded as $z\to r_j$ (cf.~\cite[\S~21]{Musk}) such that
\be
\res{z=r_j}\partial_z h_j(z)\mathrm d z=\frac{m_j}{2\pi\mathrm i}
\ee
and, by~\eqref{verylast} and~\eqref{veryverylast}, we finally get
\be
\res{z=r_j}\bigl(A(z)\bigr)_{1,2}\,\mathrm d z=-\mathrm i \e^{\mathrm i\pi\alpha}\varphi_+(r_j;x,t)=-m_jc_\alpha^2\varphi_+(r_j;x,t).
\ee
Combining these results and recalling the definition of Stieltjes integral, cf.~\eqref{eq:dsigma}, the proof of~\eqref{eq:dtlogF} is complete.
 
The proof of~\eqref{eq:dxlogF} is similar and so we only sketch it here.
We have
\begin{align}
\nonumber
&
\frac{c_\alpha^2}{2\pi}\int_{(t,+\infty)\setminus\lbrace r_1,\dots,r_k\rbrace}(\lambda- t )\,\Theta_+(\lambda)\cdot\begin{pmatrix}0 & \sigma'(\lambda) \\ 0 & 0
\end{pmatrix}\cdot\Theta_+(\lambda)^{-1}\mathrm d\lambda
\\
&\qquad=\res{z= t } (z- t )A(z)\mathrm{d} z+\res{z=\infty}(z- t )A(z)\mathrm{d} z+\sum_{j=1}^k1_{r_j>t}\res{z=r_j}(z- t )A(z)\mathrm{d} z.
\end{align}
By Proposition~\ref{prop:pawpaeta} and~\eqref{eq:diffeqt}, the residue at~$t$ vanishes.
Consider the $(1,2)$-entry and compute the residue at~$\infty$ by~\eqref{eq:Psiasympinfty} to get
\be
\res{z=\infty}(z- t )A(z)\mathrm{d} z=\frac {xt}2-x r-\frac 12v-\frac 12 xv^2=\frac 12\bigl(xt-(xv)_x\bigr)
\ee
where we use~\eqref{eq:rdef} and the relation $2r+v^2=v_x$ established just after~\eqref{eq:rdef}.
Finally, the residues at $r_j$ are computed similarly as in the previous case and combining these results yields~\eqref{eq:dxlogF}.
\end{proof}

\section{Small-\texorpdfstring{$x$}{x} asymptotics}

We first study small-$x$ asymptotics for $Y(z)$ off the real positive axis.

\begin{proposition}
\label{prop:small}
We have
\be
\bigl|\bigl(Y(z)\bigr)_{i,j}-\delta_{i,j}\bigr|=O\bigl(x^{2+2\alpha}\,(1+|z|)^{-1}\bigr),\qquad i,j=1,2,
\ee
as $x\to 0_+$ uniformly for $z\in\mathbb C$ such that $\mathrm{dist}(z,\mathbb R_+)>c(1+|z|)$, for any fixed $c>0$.
\end{proposition}
\begin{proof}
Recalling the explicit formula for $Y$ given in~\eqref{eq:Yexplicit}, we can estimate
\begin{align}
\nonumber
\bigl|\bigl(Y(z)\bigr)_{i,j}-\delta_{i,j}\bigr|&\leq \int_0^{+\infty}\biggl|\frac{\bigl((1-\mathcal K_{x,t}^\sigma)^{-1}f_i\bigr)(\lambda)\,g_j(\lambda)}{\lambda-z}\biggr|\,\mathrm d\lambda\\
&\leq \frac{\bigl\|(1-\mathcal K_{x,t}^\sigma)^{-1}f_i\bigr\|_{L^2(0,+\infty)}\,\bigl\|g_j\bigr\|_{L^2(0,+\infty)}}{c(1+|z|)}
\end{align}
for all $z\in\mathbb C$ such that $\mathrm{dist}(z,\mathbb R_+)>c(1+|z|)$, for any fixed $c>0$.
Denoting $\|\cdot\|_\infty$ the operator norm and $\|\cdot\|_1$ the trace-norm and recalling the estimates in the proof of Lemma~\ref{lemma:traceclass1} we have $\|\mathcal K_{x,t}^\sigma\|_\infty\leq\|\mathcal K_{x,t}^\sigma\|_1<1/2$ for $x$ sufficiently small hence $\|(1-\mathcal K_{x,t}^\sigma)^{-1}\|_\infty\leq (1-\|\mathcal K_{x,t}^\sigma\|_1)^{-1}=O(1)$ as~$x\to 0_+$. Therefore, it is enough to bound the $L^2(0,+\infty)$-norm of $f_i,g_j$, and so it suffices to bound the $L^2(0,+\infty)$-norm of $f_1,f_2$ (because $g_1=f_2$ and $g_2=-f_1$).
By standard asymptotic properties of Bessel functions, $\mathrm{J}_\alpha(\sqrt\lambda)=O(\lambda^{\frac \alpha 2})$ for $0<\lambda\leq 1$ and $\mathrm{J}_\alpha(\sqrt\lambda)=O(\lambda^{-1/4})$ for $\lambda\geq 1$, and, for any fixed $t$, $\sigma(r)=O\bigl(\min(1,(r-t)^{-b})\bigr)$ for all $r>t$ and for an arbitrarily large $b>0$ by Assumption~\ref{assumption}
Hence, as $x\to 0_+$,
\be
\label{eq:estimateL2f1}
\|f_1\|_{L^2(0,+\infty)}^2=O\biggl(\int_0^{x^2}\lambda^\alpha\mathrm d\lambda + \int_{x^2}^1\lambda^\alpha(x^{-2}\lambda)^{-b}\mathrm d\lambda
+\int_1^{+\infty}\lambda^{-\frac 12}(x^{-2}\lambda)^{-b}\mathrm d\lambda\biggr)=O(x^{2+2\alpha}).
\ee
Similarly, $\sqrt{\lambda}\mathrm{J}_\alpha'(\sqrt{\lambda})=O(\lambda^{\frac\alpha 2})$ for $0<\lambda\leq 1$ and $\sqrt{\lambda}\mathrm{J}_\alpha'(\sqrt{\lambda})=O(\lambda^{1/4})$ for $\lambda\geq 1$, hence we can prove in the same way that as $x\to 0_+$
\be
\label{eq:estimateL2f2}
\|f_2\|_{L^2(0,+\infty)}^2=O(x^{2+2\alpha})
\ee
and the proof is complete.
\end{proof}

\begin{proposition}\label{prop:smallx1}
We have
\be
\label{eq:integralcorollary}
Q_\sigma(x,t) = \exp\biggl(2c_\alpha^2\int_0^x\log\biggl(\frac x{x'}\biggr) \biggl(\int_t ^{+\infty}(\lambda- t)\,\varphi_+(\lambda;x',t)^2\,\mathrm d\sigma(\lambda)\biggr)\mathrm{d} x '\biggr) .
\ee
\end{proposition}
\begin{proof}
Let $p(x,t)=2c_\alpha^2\int_ t ^{+\infty}(\lambda- t )\,\varphi_+(\lambda;x,t)^2\,\mathrm d\sigma(\lambda)$.
By \eqref{eq:vpaxQ} and~\eqref{eq:dxlogF} we have $p(x,t)=\partial_x\bigl(x\,\partial_x\log Q_\sigma(x,t)\bigr)$.
Next, note that
\be
\label{eq:xlogFxtozero}
x\,\partial_x\log Q_\sigma=\lim_{z\to\infty}\bigl(z\,Y(z)\bigr)_{1,2}=O(x^{2+2\alpha})
\ee
as $x\to 0_+$, which follows from~\eqref{eq:Ydet} and Proposition~\ref{prop:small}.
Thanks to~\eqref{eq:xlogFxtozero} we get $x\,\partial_x\log Q_\sigma(x,t)= \int_0^xp(x',t)\mathrm{d} x'$.
Integrating again we have, for all $\varepsilon>0$,
\begin{align}\nonumber
\log \frac{Q_\sigma(x,t)}{Q_\sigma(\varepsilon,t)} &= \int_\varepsilon^x\int_0^{x'}p(x'',t)\mathrm{d} x''\frac{\mathrm{d} x'}{x'} \\&= \log\biggl(\frac x\varepsilon\biggr)\biggl(\bigl(y\,\partial_y\log Q_\sigma(y,t)\bigr)\big|_{y=\varepsilon}\biggr)+\int_\varepsilon^x\log\biggl(\frac{x}{x'}\biggr)p(x',t)\mathrm{d} x'
\end{align}
where in the last step we integrate by parts.
To get~\eqref{eq:integralcorollary} it finally suffices to take the limit as $\varepsilon\to 0_+$, using~\eqref{eq:xlogFxtozero} and the fact that $\log Q_\sigma(\varepsilon,t)\to 0$ as $\varepsilon\to 0_+$ (cf.~Corollary~\ref{corollary:traceclass}).
\end{proof}
\begin{proposition}\label{prop:smallx2}
As $x\to 0_+$, we have $\varphi_+(\lambda;x,t)\sim c_\alpha^{-1}\sqrt{\frac x2}\,\mathrm{J}_\alpha(x\sqrt{\lambda-t})$.
\end{proposition}
\begin{proof}
We first rewrite $\varphi(z;x,t)$ as
\begin{align}
\nonumber
&\biggl(\sqrt{\frac{x}{2\pi}},0\biggr)Y\bigl(x^2(z-t)\bigr)\Phi\bigl(x^2(z-t)\bigr)\begin{pmatrix}
1 \\ 0 
\end{pmatrix}
\\
\nonumber
&=\biggl(\sqrt{\frac{x}{2\pi}},0\biggr)\biggl(I-\int_0^{+\infty}\frac{\mathbf F(\mu)\mathbf g(\mu)^\top}{\mu-x^2(z-t)}\mathrm d\mu\biggr)\Phi\bigl(x^2(z-t)\bigr)\begin{pmatrix}
1 \\ 0 
\end{pmatrix}
\\
\label{eq:last}
&=\biggl(\sqrt{\frac{x}{2\pi}},0\biggr)\biggl(\Phi\bigl(x^2(z-t)\bigr)\begin{pmatrix}
1 \\ 0 
\end{pmatrix}+\frac{\sqrt{2\pi}}{c_\alpha}\int_0^{+\infty}\mathbf F(\mu)\sqrt{\sigma(x^{-2}\mu+t)}K_{\mathsf{Be}}\bigl(\mu,x^2(z-t)\bigr)\mathrm d\mu\biggr).
\end{align}
In the first line we use~\eqref{eq:varphi1}, in the second one we use~\eqref{eq:Yexplicit}, and in the last one we use Proposition~\ref{prop:identitiesfg}.
Hence, for all $\lambda>t$, $\varphi_+(\lambda;x,t)$ is equal to
\be
\sqrt{\frac{x}{2}}c_\alpha^{-1}\biggl(\mathrm{J}_\alpha(x\sqrt{\lambda-t})
+2\sqrt\pi\int_0^{+\infty}\bigl((1-\mathcal K_{x,t}^\sigma)^{-1}f_1\bigr)(\mu)\sqrt{\sigma(x^{-2}\mu+t)}K_{\mathsf{Be}}\bigl(\mu,x^2(\lambda-t)\bigr)\mathrm d\mu\biggr).
\ee
Defining $\mathcal E$ by
\begin{align}
\nonumber
\mathcal E&=\frac{\varphi_+(\lambda;x,t)}{\sqrt{\frac{x}{2}}c_\alpha^{-1}\mathrm{J}_\alpha(x\sqrt{\lambda-t})}-1
\\&=2\sqrt{\pi}\int_0^{+\infty}\bigl((1-\mathcal K_{x,t}^\sigma)^{-1}f_1\bigr)(\mu)\sqrt{\sigma(x^{-2}\mu+t)}\frac{K_{\mathsf{Be}}\bigl(\mu,x^2(\lambda-t)\bigr)}{\mathrm{J}_\alpha(x\sqrt{\lambda-t})}\mathrm d\mu,
\end{align}
we need to show that $|\mathcal E|\to 0$ as $x\to 0_+$.
By Cauchy--Schwarz inequality, $|\mathcal E|$ is bounded above by
\be 
2\sqrt\pi\,\biggl\|\bigl((1-\mathcal K_{x,t}^\sigma)^{-1}f_1\bigr)(\mu)\biggr\|_{L^2((0,+\infty),\mathrm d\mu)}\,\cdot\,\biggl\|\sqrt{\sigma(x^{-2}\mu+t)}\frac{K_{\mathsf{Be}}\bigl(\mu,x^2(\lambda-t)\bigr)}{\mathrm{J}_\alpha(x\sqrt{\lambda-t})}\biggr\|_{L^2((0,+\infty),\mathrm d\mu)}
\ee
The first $L^2$-norm is $o(1)$ as $x\to 0_+$ by the operator norm estimate $\|(1-\mathcal K_{x,t}^\sigma)^{-1}\|_{\infty}=O(1)$ and the $L^2$-norm estimate~\eqref{eq:estimateL2f1}, like in the proof of Proposition~\ref{prop:small}.
For the second $L^2$-norm we need the estimate for the ratio of the Bessel kernel over a Bessel function proved in Lemma~\ref{lemma:Bessel} and which implies, along with the usual estimate $\sigma(r)=O\bigl(\min(1,(r-t)^{-b})\bigr)$ for all $r>t$, for any fixed $t$ and an arbitrarily large $b$,
\begin{align}
\nonumber
&\biggl\|\sqrt{\sigma(x^{-2}\mu+t)}\frac{K_{\mathsf{Be}}\bigl(\mu,x^2(\lambda-t)\bigr)}{\mathrm{J}_\alpha(x\sqrt{\lambda-t})}\biggr\|_{L^2((0,+\infty),\mathrm d\mu)}^2\\
&\quad=O\biggl(\int_0^{x^2}\mu^\alpha\mathrm d\mu + \int_{x^2}^1\mu^\alpha(x^{-2}\mu)^{-b}\mathrm d\mu
+\int_1^{+\infty}\mu^{-\frac 12}(x^{-2}\mu)^{-b}\mathrm d\mu\biggr)
=o(1)
\end{align}
as $x\to 0_+$ and the proof is complete.
\end{proof}

\section{Proof of Theorem~\ref{thm}}\label{sec:proof}

\begin{proof}[Proof of Theorem~\ref{thm}]
\textit{(i)} follows from Lemma~\ref{lemma:traceclass1} and Corollary~\ref{corollary:traceclass}.
\textit{(ii)} follows from Proposition~\ref{prop:equations}.
About~\textit{(iii)}, $f(x;\lambda,t)$ defined by
\be
\label{eq:deff}
f(x;\lambda,t)=\sqrt 2\,c_\alpha\,\varphi_+(\lambda;x,t)
\ee
satisfies the properties stated in the theorem:~\eqref{eq:Schrodingerstatement} follows from~\eqref{eq:Schrodinger} and Proposition~\ref{prop:smallx2},~\eqref{eq:integration} follows from Proposition~\ref{prop:smallx1}, \eqref{eq:eqtstatement} follows from~\eqref{eq:diffeqt} (by looking at the $1,1$-entry), and~\eqref{eq:nonlocal} follows from Proposition~\ref{prop:nonlocal}.
It is worth noting that $f(x;\lambda,t)$ is real-valued by~\eqref{eq:conjugate}.
\end{proof}

\begin{remark}\label{remark:asymp}
The asymptotic relations~\eqref{eq:asympfintegrals} follow from the definition~\eqref{eq:deff}.
Indeed, the first relation follows from the fact that $Y_+(\lambda)\sim I$ as $\lambda\to +\infty$ and asymptotic properties of Bessel functions, recalling the definition~\eqref{eq:Psi} of $\Psi$ and~\eqref{eq:varphi1}.
The second one follows from the first point in Proposition~\ref{proposition:behavioratzero}.
\end{remark}

\section{Proof of Theorem~\ref{thm2}}\label{sec:proof2}

\begin{proof}[Proof of Theorem~\ref{thm2}]
We have
\begin{align}
\nonumber
&-\partial_x\log Q_\sigma(x,t) =-2x^{-3}\biggl[\int_0^{+\infty}\lambda\,\sigma'(x^{-2}\lambda+t)\,K_{\mathsf{Be}}(\lambda,\lambda)\,\mathrm d\lambda 
\\
\label{eq:Jacobibis}
&\ +\sum_{n\geq 1}\int_0^{+\infty}\lambda\,\sigma'(x^{-2}\lambda+t)\biggl(\int_{(0,+\infty)^n}K_{\mathsf{Be}}(\lambda,\mu_1)\cdots K_{\mathsf{Be}}(\mu_n,\lambda)\prod_{i=1}^n\sigma(x^{-2}\mu_i+t)\mathrm d\mu_i\biggr)\,\mathrm d\lambda\biggr].
\end{align}
By an integration by parts and a change of variable, the first term in the right-hand side is
\be
2x^{-1}\int_t^{+\infty}\sigma(\lambda)\frac{\mathrm d}{\mathrm d\lambda}\biggr(x^2(\lambda-t)\,\Delta_{\mathsf{Be}}\bigl(x^2(\lambda-t)\bigr)\biggr)\,\mathrm d\lambda
\ee
where we denote  $\Delta_{\mathsf{Be}}(\lambda)=K_{\mathsf{Be}}(\lambda,\lambda)$.
A direct computation shows that $\partial_\xi\bigl(\xi\Delta_{\mathsf{Be}}(\xi)\bigr)=\frac 14\mathrm{J}_\alpha(\sqrt\xi)^2$, implying that 
\be
2x^{-1}\frac{\mathrm d}{\mathrm d\lambda}\biggr(x^2(\lambda-t)\,\Delta_{\mathsf{Be}}\bigl(x^2(\lambda-t)\bigr)\biggr)=\frac x2\mathrm{J}_\alpha(x\sqrt{\lambda-t})^2.
\ee
Since $\mathrm{J}_\alpha(x\sqrt{\lambda-t})\sim x^{\alpha}\frac{2^{-\alpha}}{\Gamma(\alpha+1)}(\lambda-t)^{\alpha/2}$ as $x\to 0_+$, we have
\be
\frac x2\int_t^{+\infty}\sigma(\lambda)\mathrm{J}_\alpha(x\sqrt{\lambda-t})^2\mathrm d\lambda
\sim x^{2\alpha+1}\frac{2^{-2\alpha-1}}{\Gamma(\alpha+1)^2}
\int_t^{+\infty}\sigma(\lambda)(\lambda-t)^\alpha\mathrm d\lambda,\quad x\to 0_+,
\ee
as one can prove, for example, by Lebesgue dominated convergence theorem, using the fact that~$\sigma$ decays rapidly at~$+\infty$.
It remains to argue that the remaining terms in the right-hand side of~\eqref{eq:Jacobibis} give a negligible contribution with respect to the first one.
To this end, note that by~\eqref{eq:conv} and Cauchy inequality, we have $|K_{\mathsf{Be}}(\lambda,\mu)|^2\leq \Delta_{\mathsf{Be}}(\lambda)\Delta_{\mathsf{Be}}(\mu)$ such that each of these terms is bounded, in absolute value, by $PQ^n$ where
\be
P=2x^{-3}\int_{0}^{+\infty}\lambda\,\sigma'(x^{-2}\lambda+t)\,\Delta_{\mathsf{Be}}(\lambda)\,\mathrm d\lambda,\quad Q=\int_0^{+\infty}\sigma(x^{-2}\mu+t)\,\Delta_{\mathsf{Be}}(\mu)\,\mathrm d\mu,
\ee
and so the sum of these terms gives $PQ/(1-Q)$.
We just proved that, when $x\to 0_+$, $P=O(x^{2\alpha+1})$, while $Q=O(x^{2\alpha+2})$ follows from the arguments in the proof of Lemma~\ref{lemma:traceclass1}, cf.~\eqref{eq:finitetrace} and~\eqref{eq:last2}, and so~\eqref{eq:boundaryvalue} is proved.
\end{proof}

\section*{Acknowledgements}

I am grateful to Mattia Cafasso, Tom Claeys, Gabriel Glesner, Guilherme Silva, and Sofia Tarricone for valuable discussions.
I extend my gratitude to the anonymous referees for valuable comments and suggestions which improved the presentation.
This work is supported by the FCT grant 2022.07810.CEECIND and by the Grupo de F\'isica Matem\'atica (UIDB/00208/2020, UIDP/00208/2020, DOI: 10.54499/UIDB/00208/2020, 10.54499/UIDP/00208/2020).

%

\end{document}